\documentclass[11pt]{article}

\usepackage[utf8]{inputenc}
\usepackage[T1]{fontenc}
\usepackage{lmodern}
\usepackage[margin=1in]{geometry}
\usepackage{microtype}
\usepackage{setspace}
\usepackage{amsmath,amsthm,amssymb,mathtools}
\usepackage{bm}
\usepackage{bbm}
\usepackage{booktabs}
\usepackage{comment}
\usepackage[linesnumbered,ruled,vlined]{algorithm2e}

\SetKwInput{KwInit}{Initialize}

\usepackage{graphicx}

\usepackage{subcaption}
\usepackage{enumitem}
\usepackage[hidelinks]{hyperref}
\usepackage{multirow}
\usepackage[round]{natbib}
\usepackage{float}

\numberwithin{equation}{section}

\theoremstyle{plain}
\newtheorem{theorem}{Theorem}[section]

\newtheorem{lemma}[theorem]{Lemma}

\theoremstyle{definition}
\newtheorem{definition}[theorem]{Definition}

\theoremstyle{remark}
\newtheorem{remark}[theorem]{Remark}

\newcommand{\R}{\mathbb{R}}

\newcommand{\1}{\mathbbm{1}}

\DeclareMathOperator{\tr}{tr}
\DeclareMathOperator{\diag}{diag}
\DeclareMathOperator{\argmin}{arg\,min}

\DeclareMathOperator{\dom}{dom}
\DeclareMathOperator{\prox}{prox}

\newcommand{\norm}[1]{\left\lVert #1 \right\rVert}

\newcommand{\Symq}{\mathbb{S}^q}

\newcommand{\Symqpp}{\mathbb{S}^q_{++}}

\newcommand{\SXX}{\widehat{\Sigma}_{XX}}
\newcommand{\SYY}{\widehat{\Sigma}_{YY}}
\newcommand{\SXY}{\widehat{\Sigma}_{XY}}
\newcommand{\SYX}{\widehat{\Sigma}_{YX}}

\newcommand{\LG}{\mathcal{L}}

\title{\vspace{-1ex}Convex Reparameterization and Self-Concordant Algorithms for Multivariate Regression with Covariance Estimation}

\author{
Hongru Zhao\\[2pt]
\small School of Statistics, University of Minnesota Twin Cities\\
\small \texttt{zhao1118@umn.edu}
\and
Huiqian Feng\\[2pt]
\small Department of Applied Mathematics \& Statistics, Stony Brook University\\
\small \texttt{huiqian.feng@stonybrook.edu}
}

\date{\small\today}

\begin{document}
\maketitle
\doublespacing

\begin{abstract}
Building on the reparameterization of \citet{zhu2020convex} for multivariate linear regression, which yields a jointly convex penalized likelihood in the reparameterized regression coefficient matrix and the precision matrix, we show that the resulting scaled Gaussian loss is standard self-concordant. This places the penalized likelihood problem of jointly estimating the precision matrix and the reparameterized regression coefficient matrix within composite self-concordant optimization \citep{tran2015composite} and yields two algorithms: a proximal gradient method and a damped proximal Newton method. In simulations, we study algorithmic robustness, iterations to convergence, and elapsed time. In a protein expression application, compared with the classical-parameterization formulation, the proposed convex formulation gives similar mean squared prediction error and can be substantially faster when the fitted precision matrix is dense.
\end{abstract}

\section{Introduction}\label{sec:intro}

In many applications, such as gene and protein regulatory networks \citep{friedman2008glasso, Yin2011, Cai2013} and econometric systems \citep{demirer2018estimating, barigozzi2019nets}, one observes multiple responses whose error dependence structure is scientifically meaningful. Multivariate linear regression simultaneously models these response variables in terms of one or more predictors. In such settings, penalized Gaussian likelihood provides a way to estimate regression effects and the error precision matrix jointly, often under structural constraints such as sparsity.

Concretely, let $x_i^{\mathrm{raw}}\in\R^{p}$ and $y_i^{\mathrm{raw}}\in\R^{q}$ denote the raw predictors and responses. Consider the multivariate linear regression model with Gaussian errors \(y_i^{\mathrm{raw}}\mid x_i^{\mathrm{raw}}\sim\mathcal{N}_q \bigl(\alpha + (\beta^{\mathrm{raw}})^\top x_i^{\mathrm{raw}},\Sigma\bigr),\) with intercept $\alpha\in\R^q$, slope matrix on the raw scale $\beta^{\mathrm{raw}}\in\R^{p\times q}$, and precision matrix $\Omega=\Sigma^{-1}$. In practice, estimation is carried out on centered (and, if desired, standardized) predictors and centered responses; on this working scale we write $x_i$ and $y_i$ and fit \( y_i\mid x_i\sim\mathcal{N}_q\bigl(\beta^\top x_i,\Sigma\bigr),\) 
where $\beta\in\R^{p\times q}$ is the slope matrix on the preprocessed scale.

Within the classical parameterization $(\beta,\Omega)$, this model and related estimators have been studied extensively \citep[see, e.g.,][]{rothman2010,Cai2013,Yin2011,Rai2012,Lee2012,Chen2017}. Several estimation strategies are available. One common two-stage strategy estimates $\beta$ first and then estimates a sparse precision matrix from the residuals, for example by the graphical lasso or constrained $\ell_1$ minimization for inverse matrix estimation (CLIME) \citep{friedman2008glasso,cai2011clime}. Conditional-likelihood and conditional Gaussian graphical model methods instead model each response given the predictors and the remaining responses \citep{Yin2011,Lee2012,wang2015joint,wytock2013sparse,sohn2012joint}. Related extensions include multiresponse models that leverage output or task structure \citep{Rai2012,Chen2017}, multivariate linear mixed effects models \citep{navon2020capturing}, and robust multivariate lasso regression with covariance estimation \citep{chang2023robust}. Other approaches impose more specific response-dependence structure, such as latent factors \citep{zhou2017sparse,chan2025drfarm} or dense equicorrelation/compound-symmetry covariance structure for strongly associated responses \citep{ham2025sparse}. A central joint likelihood benchmark is the multivariate regression with covariance estimation (MRCE) method of \citet{rothman2010}, which estimates $(\beta,\Omega)$ jointly with $\ell_1$ penalties on the regression coefficients and the off-diagonal precision entries. In the classical parameterization $(\beta,\Omega)$, however, the joint objective is generally nonconvex, so local optimization methods can converge to different stationary points under different initial values.

To address this nonconvexity, following \citet{zhu2020convex} we reparameterize the model through the Gaussian natural parameter $\gamma=\beta\Omega$ and optimize over $(\gamma,\Omega)$, for which the negative log-likelihood is jointly convex. This transformation also has a natural interpretation from the viewpoint of conditional Gaussian graphical models: for each response coordinate, the column $\gamma_{\cdot j}$ describes predictor--response links after adjusting for the remaining responses, whereas $\beta=\gamma\Omega^{-1}$ governs the conditional mean and can combine direct predictor effects with effects transmitted through the response dependencies. As a result, sparsity in $\gamma$ targets conditional predictor--response associations rather than mean regression coefficients alone; we return to this point in Subsection~\ref{subsec:reparam_loss}. We therefore refer to $\gamma$ as the predictor--response link parameter throughout the paper.

On the preprocessed training data this leads to the following composite convex optimization problem:
\begin{equation}\label{equ:opt_main_problem}
(\hat\gamma,\hat\Omega)\in\arg\min_{\gamma\in\mathbb{R}^{p\times q},\ \Omega\in\mathbb{S}_{++}^q}
\;\Bigl\{ {\LG}(\gamma,\Omega)+\lambda_\gamma \|\gamma\|_1 + \lambda_\Omega \|\Omega\|_{1,\mathrm{off}} \Bigr\},
\end{equation}
where $\mathbb{S}_{++}^q$ denotes the cone of $q\times q$ symmetric positive-definite matrices, and $\LG(\gamma,\Omega)$ is the scaled reparameterized Gaussian negative log-likelihood, written explicitly in Section~\ref{sec:self_concordant}. For any matrix $A=(a_{ij})$, we write
\( \|A\|_{\max}:=\max_{i,j}|a_{ij}|,\ \|A\|_{\max,\mathrm{off}}:=\max_{i\ne j}|a_{ij}|,\ \|A\|_{1,\mathrm{off}}:=\sum_{i\ne j}|a_{ij}|.\)
Once $(\hat\gamma,\hat\Omega)$ is obtained, Lemma~\ref{lem:centering_intercept} recovers the slope and intercept on the raw scale.
Although the reparameterization yields joint convexity, the geometry of the smooth loss is not governed by a fixed Euclidean geometry: because the likelihood involves $\Omega^{-1}$, its gradient is not globally Lipschitz in a fixed Euclidean norm. As a result, classical fixed-step analyses and step-size choices based on $L$-smoothness are not directly applicable.

To handle this geometry, we use the second and third derivatives of the loss. We prove that the scaled loss $\LG$ is standard self-concordant. This allows us to use composite self-concordant theory \citep{tran2015composite} for a proximal gradient method with a local curvature correction and a damped proximal Newton method with the usual Newton-decrement step length. We refer to the resulting estimator as \textsc{SCA-MRCE} (Self-Concordant Algorithms for Multivariate Regression with Covariance Estimation). In the numerical studies, the proximal gradient variants are more stable and faster along a penalty path, while proximal Newton is more attractive when a more accurate solution is needed. We also study aggressive variants that reduce inner computational work and elapsed time, but they do not have the same theoretical convergence guarantees.

\subsection{A one-dimensional motivating example}\label{sec:motivating_example}
To illustrate the nonconvex optimization landscape in the classical mean and precision parameterization, consider a simple one-dimensional Gaussian setting. Let $\bar{x}\in\R$ and $S\ge 0$ denote the sample mean and empirical variance. For a mean parameter $\mu\in\R$ and a precision parameter $\Omega>0$, consider the penalized negative log-likelihood:
\begin{equation}
-\log \Omega\;+\; S\,\Omega \;+\; \Omega\,(\mu-\bar{x})^{2} \;+\; \lambda_{1}\,|\mu| \;+\; \lambda_{2}\,\Omega, \qquad (\mu,\Omega)\in\R\times(0,\infty).
\end{equation}
For suitable choices of $(\bar{x},S,\lambda_{1},\lambda_{2})$, this objective exhibits multiple local minima, as illustrated in Figure~\ref{fig:pen_gauss_1d}. This phenomenon persists in higher dimensions when working in $(\beta,\Omega)$, whereas the natural parameter reparameterization $(\gamma,\Omega)$ leads to a globally convex landscape.

\begin{figure}[htbp]
\centering
\begin{subfigure}{0.455\textwidth}
  \includegraphics[width=\linewidth]{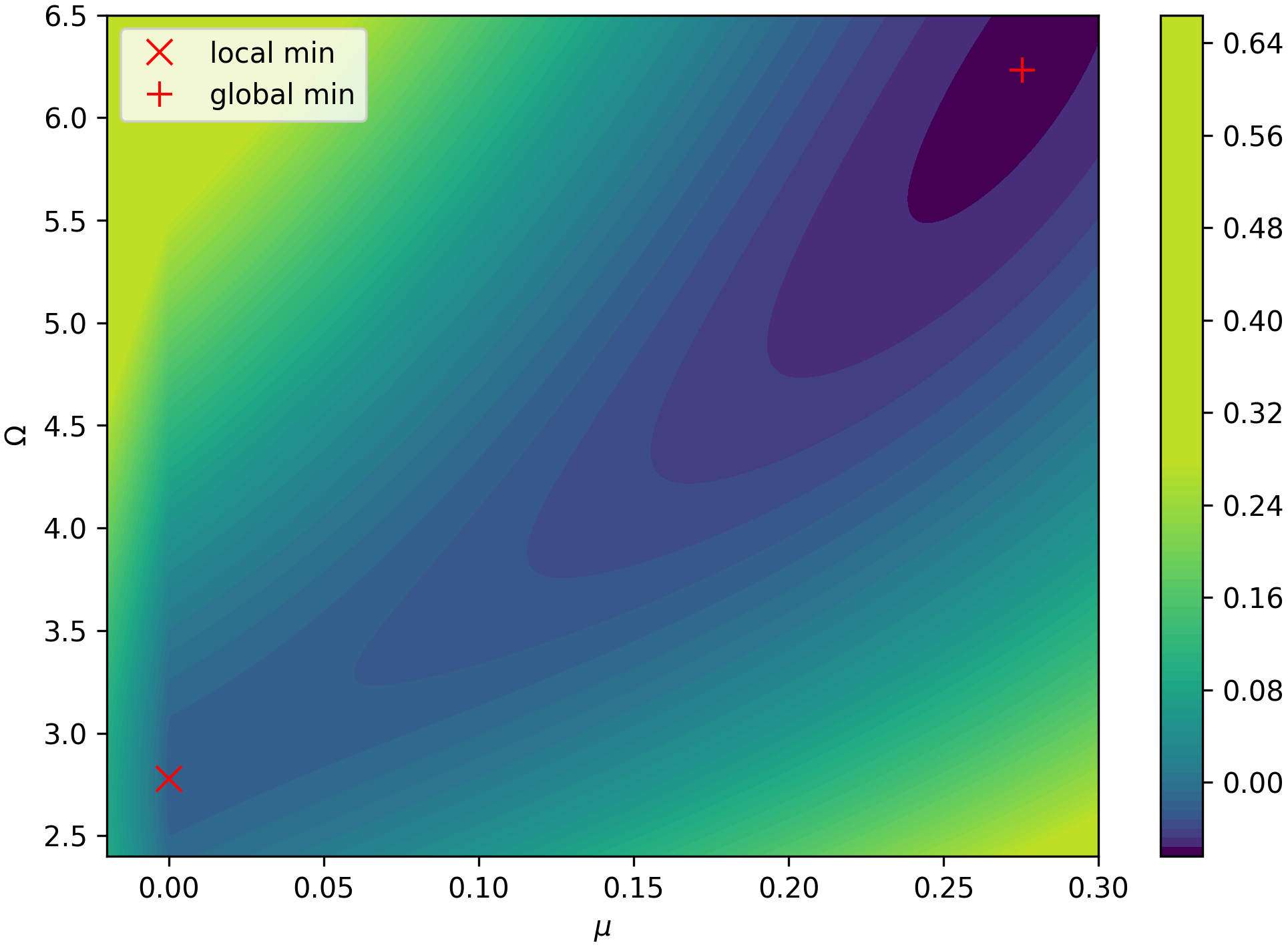}
  \caption{{Global landscape of the objective.}}
\end{subfigure}\hfill
\begin{subfigure}{0.48\textwidth}
  \includegraphics[width=\linewidth]{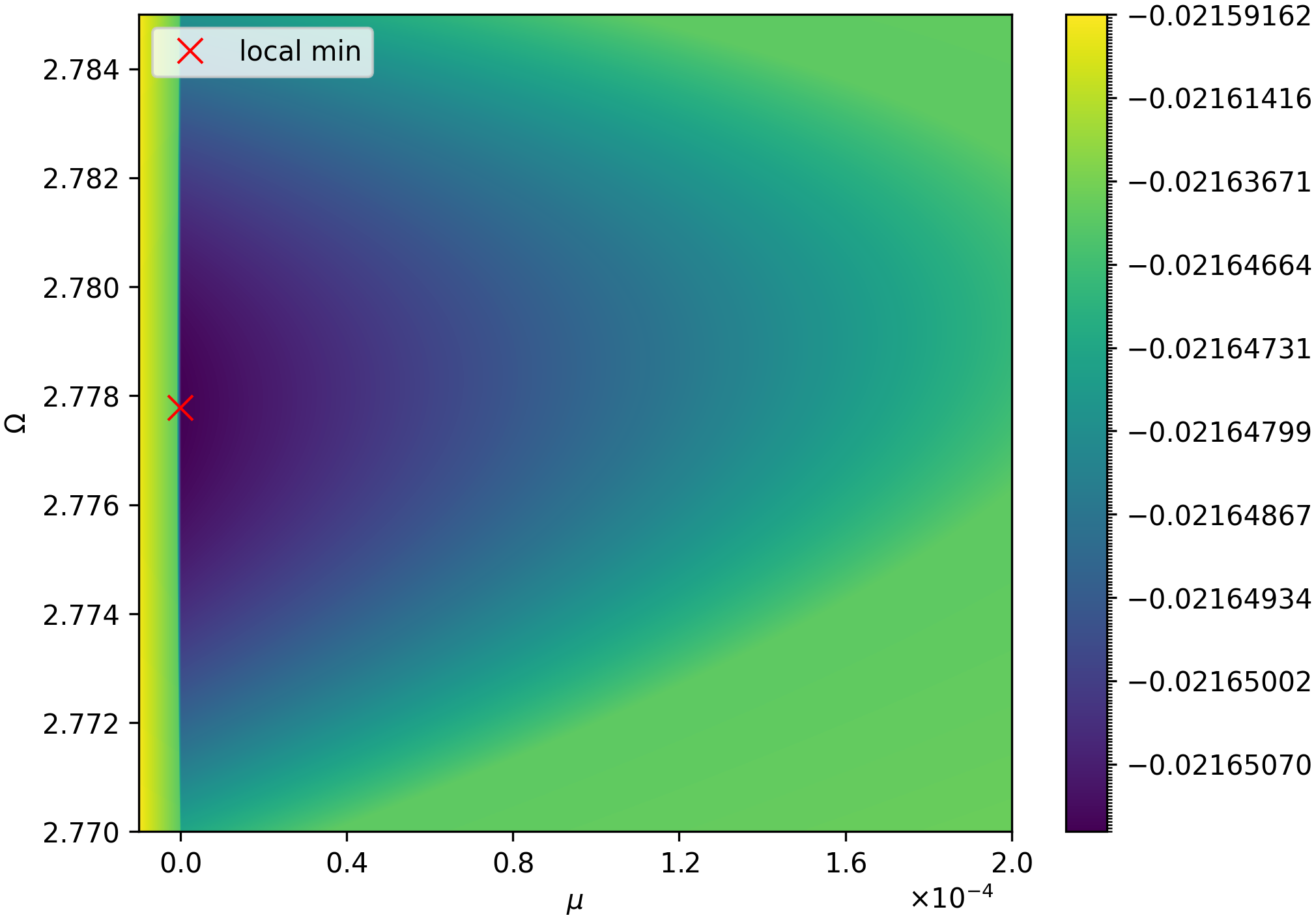}
  \caption{{Zoom near a local minimum.}}
\end{subfigure}
\caption{{Penalized univariate Gaussian example with $\bar{x}=0.5$, $S=0.01$, $\lambda_1=2.8$, and $\lambda_2=0.10$. The classical $(\mu,\Omega)$ parameterization yields multiple local minima.}}
\label{fig:pen_gauss_1d}
\end{figure}

\subsection{Contributions and outline}
Our main contributions are as follows:
\begin{itemize}[leftmargin=*]
\item \textbf{Self-concordant geometry:} We show that the scaled convex Gaussian loss $\LG$ in the $(\gamma,\Omega)$ parameterization is standard self-concordant. This places the likelihood within the composite self-concordant theory \citep{tran2015composite} and provides the local norm control used in the convergence analysis.
\item \textbf{Globally convergent iterative algorithms:} We develop self-concordant first- and second-order methods for composite penalized estimation. The proximal gradient scheme uses a local curvature correction from composite self-concordant theory, while the proximal Newton scheme uses the Newton decrement to obtain a damped step with global convergence and local quadratic convergence.
\item \textbf{Empirical evaluation:} The experiments compare predictive performance, algorithmic robustness, outer iteration counts, and path computation time across legacy and aggressive solver variants. The aggressive variants truncate expensive inner updates and reduce computation time, but they do not have the same convergence guarantees. The proximal gradient methods are the most stable and fastest family for path fitting, while proximal Newton uses many fewer outer iterations and remains useful when a more accurate solve is worth the additional inner work.
\end{itemize}

\paragraph{Organization.}
The rest of the paper is organized as follows. Section~\ref{sec:self_concordant} formulates the working scale and raw-scale models, records the derivatives and local geometry, and establishes self-concordance. Section~\ref{sec:algorithms} presents the proximal gradient and proximal Newton algorithms, the alternating direction method of multipliers (ADMM) subproblem solver, the penalty path construction, and the raw-scale recovery and evaluation metrics used later in the paper. Section~\ref{sec:sim} reports the simulation study. Section~\ref{sec:realdata} presents the real data analysis on the mice protein expression data. Section~\ref{sec:conclusion} concludes, and the appendices collect proofs and implementation details.

\section{Model Formulation and Local Geometry}
\label{sec:self_concordant}

Section~\ref{sec:intro} introduced the multivariate linear regression model with Gaussian errors and the penalized estimation problem. This section first records the predictor preprocessing and intercept recovery linking the raw and working scales, then explains the statistical interpretation of the reparameterized loss, and finally develops the self-concordant geometry used by the first- and second-order algorithms in Section~\ref{sec:algorithms}.

\subsection{Predictor preprocessing (centering/standardization) and intercept recovery}
\label{subsec:centering_intercept}

We distinguish the raw data scale from the working scale used for estimation. On the raw scale, the conditional mean model is
\(\mathbb{E}\!\left(y_i^{\mathrm{raw}}\mid x_i^{\mathrm{raw}}\right)
=
\alpha+(\beta^{\mathrm{raw}})^\top x_i^{\mathrm{raw}}.\)
Let \(X^{\mathrm{raw}}=[x_1^{\mathrm{raw}},\cdots, x_n^{\mathrm{raw}}]^{\top}\in\R^{n\times p},\ Y^{\mathrm{raw}}=[ y_1^{\mathrm{raw}},\cdots, y_n^{\mathrm{raw}}]^{\top}\in\R^{n\times q},\) and let \( \bar x = \frac{1}{n}\1_n^\top X^{\mathrm{raw}}\in\R^{1\times p},\ \bar y = \frac{1}{n}\1_n^\top Y^{\mathrm{raw}}\in\R^{1\times q}.\)
Here and below, $\1_m$ denotes the length-$m$ vector of ones. For centering only, we take $D_x=I_p$, where $I_p$ is the $p\times p$ identity matrix. For predictor standardization, we take \(D_x=\diag(s_{x,1},\dots,s_{x,p}),\)
where $s_{x,j}$ is the sample standard deviation in the training data of predictor $j$. We then define \(X := \bigl(X^{\mathrm{raw}}-\1_n\bar x\bigr)D_x^{-1}\) and \(Y := Y^{\mathrm{raw}}-\1_n\bar y.\)
We continue to write $x_i$ and $y_i$ for the corresponding observations on the working scale. Thus the training matrix after predictor preprocessing $X$ is obtained by centering and, if desired, standardizing the raw predictors.

On this working scale we fit the model without an intercept \( y_i\mid x_i\sim\mathcal{N}_q\bigl(\beta^\top x_i,\Sigma\bigr),\)
where $\beta\in\R^{p\times q}$ is the slope matrix on the preprocessed scale. The corresponding slope on the raw scale is
\(\beta^{\mathrm{raw}}=D_x^{-1}\beta.\)

Lemma~\ref{lem:centering_intercept} formalizes the equivalence between fitting on the raw scale with an unpenalized intercept and fitting on the working scale without an explicit intercept.

\begin{lemma}[Intercept recovery]\label{lem:centering_intercept}
With $D_x$, $X$, and $Y$ defined above, consider the Gaussian loss on the raw scale
\begin{align*}
\LG_{\mathrm{raw}}(\alpha,\beta^{\mathrm{raw}},\Omega)
&:=
-\log\det(\Omega) \\
&\quad
+\frac{1}{n}\sum_{i=1}^n
\bigl(y_i^{\mathrm{raw}}-\alpha-(\beta^{\mathrm{raw}})^\top x_i^{\mathrm{raw}}\bigr)^\top
\Omega
\bigl(y_i^{\mathrm{raw}}-\alpha-(\beta^{\mathrm{raw}})^\top x_i^{\mathrm{raw}}\bigr),
\quad
\Omega\succ0.
\end{align*}
For every fixed $(\beta^{\mathrm{raw}},\Omega)$, the unique minimizer in $\alpha$ is \(\alpha^\star(\beta^{\mathrm{raw}})=(\bar y-\bar x\,\beta^{\mathrm{raw}})^\top.\)
If $\beta=D_x\beta^{\mathrm{raw}}$, then
\(\LG_{\mathrm{raw}}\bigl(\alpha^\star(\beta^{\mathrm{raw}}),\beta^{\mathrm{raw}},\Omega\bigr)=\LG_{\mathrm{prep}}(\beta,\Omega),\)
where \( 
\LG_{\mathrm{prep}}(\beta,\Omega)
:=
-\log\det(\Omega)
+\frac{1}{n}\sum_{i=1}^n
\bigl(y_i-\beta^\top x_i\bigr)^\top
\Omega
\bigl(y_i-\beta^\top x_i\bigr).\)
Therefore the problem on the raw scale with an unpenalized intercept is equivalent, after this preprocessing, to the corresponding problem on $(X,Y)$ with no explicit intercept. 
\end{lemma}

Accordingly, throughout the remainder of the paper, $X$ denotes the preprocessed training design matrix, $Y$ denotes the centered training responses, and the optimization variables contain no explicit intercept. Raw-scale coefficients are recovered after fitting by Lemma~\ref{lem:centering_intercept}. An elementary proof is given in Appendix~\ref{app:proof_centering}.

\subsection{Reparameterized Gaussian loss}\label{subsec:reparam_loss}

For the training matrix after predictor preprocessing $X$ and the centered response matrix $Y$, let the empirical second moment matrices be \(\SXX := \frac{1}{n} X^\top X,\ \SYY := \frac{1}{n} Y^\top Y,\ \SYX := \frac{1}{n} Y^\top X,\ \SXY := \SYX^\top = \frac{1}{n} X^\top Y.\)
Recall the natural parameter reparameterization \(\gamma = \beta\Omega\) for the slope matrix \(\beta\) associated with the preprocessed predictor matrix \(X\). Writing
\(\xi := (\gamma,\Omega)\), the reparameterized averaged Gaussian negative log-likelihood is
\begin{equation*}
  \widetilde{\LG}(\xi)
  =
  \frac12\Bigl\{
    -\log\det(\Omega)
    + \tr\bigl(\Omega^{-1}\gamma^\top \SXX \gamma\bigr)
    - 2\tr\bigl(\SXY^\top \gamma\bigr)
    + \tr\bigl(\Omega \SYY\bigr)
  \Bigr\},
  \
  \dom(\widetilde{\LG})=\R^{p\times q}\times\Symqpp.
\end{equation*}
\noindent
For algorithmic development we will work with the scaled loss
\begin{equation*}
\LG(\xi)
:= 2\,\widetilde{\LG}(\xi)
=
-\log\det(\Omega)
+ \tr\bigl(\Omega^{-1}\gamma^\top \SXX \gamma\bigr)
- 2\tr\bigl(\SXY^\top \gamma\bigr)
+ \tr\bigl(\Omega \SYY\bigr),
\
\dom(\LG)=\dom(\widetilde{\LG}).
\end{equation*}
This scaling is convenient because it makes $\LG$ a standard self-concordant function
(Lemma~\ref{lem:self_concordant_Gauss}). It is trivial to verify that $\LG_{\mathrm{prep}}$ in Lemma~\ref{lem:centering_intercept} is identical to $\LG$.

Beyond yielding joint convexity, the reparameterization $\gamma=\beta\Omega$ also has a natural statistical interpretation. For each response coordinate $j$, standard Gaussian conditioning gives
\begin{equation}\label{eq:nodewise_conditional_gamma}
Y_j \mid X=x,\;Y_{-j}=y_{-j}
\sim
\mathcal N\!\left(
\frac{x^\top\gamma_{\cdot j}-\Omega_{j,-j}y_{-j}}{\Omega_{jj}},
\Omega_{jj}^{-1}
\right),
\end{equation}
where $\gamma_{\cdot j}$ is the $j$-th column of $\gamma$ and $\Omega_{j,-j}$ is the $j$-th row of $\Omega$ excluding its diagonal entry $\Omega_{jj}$. Hence $\gamma_{\cdot j}/\Omega_{jj}$ is the coefficient of $X$ in the conditional regression of $Y_j$ on $(X,Y_{-j})$; it measures predictor--response association after adjusting for the remaining responses. Since $\Omega_{jj}>0$, this conditional coefficient has the same support as $\gamma_{\cdot j}$, so an $\ell_1$ penalty on $\gamma$ targets sparsity in these conditional predictor--response links. By contrast, $\beta=\gamma\Omega^{-1}$ determines the conditional mean $\mathbb E(Y\mid X=x)=\beta^\top x$, and \(\beta_{\cdot j}=\sum_{k=1}^q \gamma_{\cdot k}(\Omega^{-1})_{kj},\)
so the mean coefficient for response $j$ can combine conditional links through the response network; in particular, $\beta_{\cdot j}$ may be nonzero even when $\gamma_{\cdot j}=0$. This interpretation is statistical rather than causal and is standard in the conditional Gaussian graphical model (CGGM) literature \citep{zhu2020convex,chiquet2017structured,sohn2012joint}. Meanwhile, penalizing the off-diagonal entries of $\Omega$ targets sparsity in the response--response conditional dependence network \citep{wang2015joint,wytock2013sparse,gan2022bayesian}.

We next record the gradient, Hessian, and local norms needed for the optimization analysis of the preprocessed objective.

\subsection{Derivatives, Hessian, and local norms}
We work on the product space (ambient space) \(\Xi := \R^{p\times q}\times \Symq\), with generic increment \(d\xi=(d\gamma,d\Omega)\in\Xi\), where $\Symq$ denotes the space of symmetric $q\times q$ matrices. This space is equipped with the Frobenius inner product
\(
  \langle d\xi_1,d\xi_2\rangle
  :=
  \tr(d\gamma_1^\top d\gamma_2)+\tr(d\Omega_1^\top d\Omega_2).
\)
When convenient, we identify \(\xi=(\gamma,\Omega)\) with the stacked matrix
\(
\xi= (\gamma^\top \ \Omega^\top)^\top
\in\R^{(p+q)\times q}
\)
(and similarly for increments \(d\xi\)).

The gradients of \(\LG\) with respect to \(\gamma\) and \(\Omega\) are $\nabla_{\gamma}\,\LG(\xi)
= -\,2\SXY + 2\SXX\,\gamma\,\Omega^{-1},
\
\nabla_{\Omega}\,\LG(\xi)
= \SYY - \Omega^{-1} - \Omega^{-1}\gamma^\top \SXX\gamma\,\Omega^{-1}.$ We write \(\nabla \LG(\xi) := \bigl(\nabla_\gamma\LG(\xi),\nabla_\Omega\LG(\xi)\bigr)\in\Xi\).

For any increment \(d\xi=(d\gamma,d\Omega)\in\Xi\), the Hessian quadratic form is
\(\nabla^2\LG [\xi](d\xi,d\xi)=\tr\bigl( d\xi^\top H(\xi)\,d\xi\,\Omega^{-1}\bigr).\)
The block matrix \(H(\xi)\in\R^{(p+q)\times(p+q)}\) is
\[
H(\xi)=
\begin{pmatrix}
2\SXX & -\,2\SXX\,\gamma\,\Omega^{-1}\\[4pt]
-\bigl(2\SXX\,\gamma\,\Omega^{-1}\bigr)^\top &
\Omega^{-1}+2\,\Omega^{-1}\,\gamma^\top\,\SXX\,\gamma\,\Omega^{-1}
\end{pmatrix}.
\]
Equivalently, \(\nabla^2\LG[\xi]\) can be viewed as a self-adjoint linear map on \(\Xi\)
through
\(
  \nabla^2\LG[\xi](d\xi)
  =
  \bigl(H(\xi)d\xi\,\Omega^{-1}\bigr)_{\Xi},
\)
with the quadratic form given by \(\langle d\xi,\nabla^2\LG[\xi](d\xi)\rangle\).

The Hessian induces the local (semi-)norm at \(\xi\),
\[
\norm{d\xi}_{\xi}
:=
\sqrt{\nabla^2\LG[\xi](d\xi,d\xi)}
=
\sqrt{\langle d\xi,\;H(\xi)d\xi\,\Omega^{-1}\rangle}.
\]
Because \(\SXX\) may be singular in high dimensions, \(\|\cdot\|_{\xi}\) is in general only a seminorm.

\subsection{Self-concordance}
Self-concordance replaces global Lipschitz gradient assumptions with a local third derivative
inequality that is well suited to objectives involving \(\log\det(\Omega)\) and \(\Omega^{-1}\).

\begin{definition}[Self-concordant function]\label{def:self_concordant_1d}
Let $V$ be a finite-dimensional real vector space and let $C \subseteq V$ be
open and convex. A three times continuously differentiable function
$f \colon C \to \R$ is called {$M$-self-concordant} for some constant
$M>0$ if, for every $x \in C$ and every direction $h \in V$, the
one-dimensional restriction
\(\phi(t) := f(x + t h)\),
defined for all $t$ such that $x + t h \in C$, satisfies
\(
  \bigl| \phi'''(0) \bigr|
  \;\le\;
  M \, \bigl( \phi''(0) \bigr)^{3/2}.
\)
The special case $M=2$ is often called a standard self-concordant function.
\end{definition}

\begin{lemma}[Self-concordance of the reparameterized Gaussian loss]\label{lem:self_concordant_Gauss}
The scaled loss $\LG$ 
is standard self-concordant. 
\end{lemma}

\noindent
The proof is given in Appendix~\ref{app:proof_self_concordant}.
With standard self-concordance in place for the scaled objective, we can invoke classical composite self-concordant theory \citep{tran2015composite} to justify damped proximal Newton updates with guaranteed global convergence and fast local rates.

Lemma~\ref{lem:self_concordant_Gauss} also changes the computational regime for computing the penalized estimator. Statistically, the curvature of the Gaussian likelihood is governed by the current precision structure and is highly nonuniform on the positive-definite cone, especially near the boundary where the \(\log\det\) term acts as a barrier. Thus a single global Lipschitz constant is typically too crude to reflect the local geometry induced by \(\Omega\), whereas self-concordance provides a model-intrinsic local metric. Computationally, step-size selection can therefore use local geometric quantities rather than repeated objective evaluations. For proximal-gradient updates, composite self-concordant theory replaces generic Armijo backtracking with a local curvature correction based on quantities such as the local decrement, thereby avoiding expensive function-evaluation line searches \citep{tran2015composite}. This distinction is practically important here because evaluating the smooth loss requires computationally demanding \(\mathcal{O}(q^3)\) matrix operations associated with \(\log\det(\Omega)\) and \(\Omega^{-1}\).

The implications for second-order methods are even more substantial. Self-concordance yields a principled damped proximal Newton step length of the form \(\alpha=(1+\nu)^{-1}\), where \(\nu\) is the local Newton decrement. As shown in Section~\ref{sec:algorithms}, this choice guarantees monotone descent and global convergence without a separate globalization or heuristic backtracking procedure. Once the iterates enter the local Newton region, \(\nu \to 0\), so \(\alpha \to 1\) and the method transitions to full Newton steps with quadratic local convergence. We next develop a self-concordant proximal gradient baseline, whose step-size correction uses only local geometry, and a damped proximal Newton method, which makes fuller use of curvature while retaining rigorous global guarantees.

\section{Algorithms}
\label{sec:algorithms}

By Lemma~\ref{lem:centering_intercept}, all algorithmic statements below are written for preprocessed training data. The intercept has already been profiled out, so the optimization variables are $(\gamma,\Omega)$ only. We solve the composite convex problem
\begin{equation}
\label{eq:composite_obj}
\min_{\xi\in\mathcal{D}_\delta}\;
F(\xi)
:=
\LG(\xi)+\Phi(\xi),
\qquad
\mathcal{D}_\delta
:=
\bigl\{(\gamma,\Omega)\in\R^{p\times q}\times\Symqpp:\ \Omega\succeq \delta I_q\bigr\},
\end{equation}
with safeguarding constant $\delta>0$ and separable penalty $\Phi(\gamma,\Omega)=\lambda_\gamma\|\gamma\|_1+\lambda_\Omega\|\Omega\|_{1,\mathrm{off}}.$
The first- and second-order schemes below are specialized from the composite self-concordant framework of \citet{tran2015composite}. Because of the nonsmooth penalty, the algorithms use an outer loop and an inner iterative solver for the penalized subproblem; some steps add a nested loop. We refer to the main updates of the objective parameters, such as the proximal gradient and proximal Newton steps, as the outer iterations.

\subsection{Backtracking proximal gradient method}

The proximal gradient (PG) scheme below is the composite self-concordant first-order method of \citet{tran2015composite} specialized to $F=\LG+\Phi$. Given $\xi^{(m)}$, the proximal gradient trial point solves
\begin{equation}
\label{eq:prox_grad_subproblem}
s^{(m)}
=
\argmin_{\xi\in\mathcal{D}_\delta}
\Bigl\{
\langle\nabla\LG(\xi^{(m)}),\xi-\xi^{(m)}\rangle
+\frac{L^{(m)}}{2}\|\xi-\xi^{(m)}\|_F^2
+\Phi(\xi)
\Bigr\}.
\end{equation}
Solving \eqref{eq:prox_grad_subproblem} reduces to two proximal operators: the $\gamma$ block is entrywise soft-thresholding, and the $\Omega$ block is a positive-definite proximal update based on soft-thresholding of the off-diagonal entries; see Appendix~\ref{app:omega_prox} and \citet{xue2012positive}. Writing $d^{(m)}:=s^{(m)}-\xi^{(m)}$, we use the composite self-concordant step length
\[
\alpha^{(m)}
=
\frac{\beta_{(m)}^2}{\nu_{(m)}\bigl(\nu_{(m)}+\beta_{(m)}^2\bigr)},
\qquad
\nu_{(m)}:=\|d^{(m)}\|_{\xi^{(m)}},
\qquad
\beta_{(m)}:=\sqrt{L^{(m)}}\,\|d^{(m)}\|_F.
\]

\begin{algorithm}[H]
\caption{Backtracking proximal gradient for \eqref{eq:composite_obj}}
\label{alg:proxgrad}
\KwIn{Preprocessed training data $(X,Y)$; penalties $\lambda_\gamma,\lambda_\Omega$; safeguard $\delta>0$; tolerance $\varepsilon>0$.}
\KwInit{Choose $\xi^{(0)}\in\mathcal{D}_\delta$ and $L^{(0)}>0$.}
\For{$m=0,1,2,\dots$}{
  Set $L^{(m)} \leftarrow L^{(m-1)}$ (with $L^{(-1)}:=L^{(0)}$)\;
  \While{\textnormal{True}}{
    Compute $s^{(m)}$ from \eqref{eq:prox_grad_subproblem} and $d^{(m)}=s^{(m)}-\xi^{(m)}$\;
    Compute $\nu_{(m)}=\|d^{(m)}\|_{\xi^{(m)}}$ and $\beta_{(m)}=\sqrt{L^{(m)}}\|d^{(m)}\|_F$\;
    \eIf{$\nu_{(m)}^{2}/\beta_{(m)}^{2}+\nu_{(m)} > 1$}{
      $L^{(m)} \leftarrow 2L^{(m)}$\;
    }{
      \textbf{break}\;
    }
  }
  Set $\alpha^{(m)} \leftarrow \beta_{(m)}^{2}/\{\nu_{(m)}(\nu_{(m)}+\beta_{(m)}^{2})\}$\;
  Update $\xi^{(m+1)} \leftarrow \xi^{(m)} + \alpha^{(m)} d^{(m)}$\;
  \If{$\|d^{(m)}\|_F \le \varepsilon$}{
    \textbf{break}\;
  }
}
\end{algorithm}

\subsection{Damped proximal Newton method}

The proximal Newton (PN) scheme below is the composite self-concordant second-order method of \citet{tran2015composite} specialized to $F=\LG+\Phi$. At iterate $\xi^{(m)}$, the proximal Newton subproblem is
\begin{equation}
\label{eq:prox_newton_subproblem}
s^{(m)}
=
\argmin_{\xi\in\mathcal{D}_\delta}
\Bigl\{
\langle\nabla\LG(\xi^{(m)}),\xi-\xi^{(m)}\rangle
+\frac12\nabla^2\LG[\xi^{(m)}]\bigl(\xi-\xi^{(m)},\xi-\xi^{(m)}\bigr)
+\Phi(\xi)
\Bigr\}.
\end{equation}
With $\Sigma^{(m)}:=(\Omega^{(m)})^{-1}$ and stacked representation of $\xi$, this is equivalently
\begin{equation}
\label{eq:pn_matrix_form}
s^{(m)}
=
\argmin_{\xi\in\mathcal{D}_\delta}
\Bigl\{
\frac12 \tr\bigl(\xi^\top H^{(m)}\xi\,\Sigma^{(m)}\bigr)
+ \tr\bigl(\xi^\top G^{(m)}\bigr)
+ \Phi(\xi)
\Bigr\},
\end{equation}
where $H^{(m)}:=H(\xi^{(m)})$ and
\(G^{(m)}
:=
\nabla\LG(\xi^{(m)})-H^{(m)}\xi^{(m)}\Sigma^{(m)}.\)
The direction $d^{(m)}:=s^{(m)}-\xi^{(m)}$ is damped with the standard self-concordant step length
\(\alpha^{(m)}
=
(1+\nu_{(m)})^{-1},
\
\nu_{(m)}:=\|d^{(m)}\|_{\xi^{(m)}}.\)

\begin{algorithm}[H]
\caption{Damped proximal Newton for \eqref{eq:composite_obj}}
\label{alg:proxnewton}
\KwIn{Preprocessed training data $(X,Y)$; regularization parameters $\lambda_\gamma,\lambda_\Omega$; safeguard $\delta>0$; tolerances $\varepsilon_{\mathrm{step}},\varepsilon_{\mathrm{newton}}>0$.}
\KwInit{Choose $\xi^{(0)}\in\mathcal{D}_\delta$.}
\For{$m=0,1,2,\dots$}{
  Approximately solve \eqref{eq:pn_matrix_form} by Algorithm~\ref{alg:admm_pn} to obtain $s^{(m)}$\;
  Set $d^{(m)}\leftarrow s^{(m)}-\xi^{(m)}$ and $\nu_{(m)}\leftarrow \|d^{(m)}\|_{\xi^{(m)}}$\;
  Set $\alpha^{(m)}\leftarrow (1+\nu_{(m)})^{-1}$\;
  Update $\xi^{(m+1)} \leftarrow \xi^{(m)}+\alpha^{(m)}d^{(m)}$\;
  \If{$\|d^{(m)}\|_F \le \varepsilon_{\mathrm{step}}$ \textbf{or} $\nu_{(m)}\le \varepsilon_{\mathrm{newton}}$}{
    \textbf{break}\;
  }
}
\end{algorithm}

\subsection{ADMM solver for the proximal Newton subproblem}
\label{subsec:admm_pn}

We solve the PN subproblem approximately by the alternating direction method of multipliers (ADMM). A closely related splitting appears in the supplementary material of \citet{zhu2020convex}, and we adapt it to our self-concordant PN framework. Introducing a copy variable $\zeta$ gives
\begin{equation}
\label{eq:pn_splitting}
\min_{\xi\in\R^{(p+q)\times q},\,\zeta\in\mathcal{D}_\delta}\;
\frac12 \tr\bigl(\xi^\top H^{(m)}\xi\,\Sigma^{(m)}\bigr)
+ \tr\bigl(\xi^\top G^{(m)}\bigr)
+ \Phi(\zeta),
\qquad \text{s.t. }\ \xi=\zeta.
\end{equation}
With scaled dual variable $U$ and penalty parameter $\rho>0$, the ADMM updates are
\begin{align}
\label{eq:admm_xi}
\xi^{(k+1)}
&=
\argmin_{\xi}\;
\frac12 \tr\bigl(\xi^\top H^{(m)}\xi\,\Sigma^{(m)}\bigr)
+ \tr\bigl(\xi^\top G^{(m)}\bigr)
+\frac{\rho}{2}\|\xi-\zeta^{(k)}+U^{(k)}\|_F^2,\\
\label{eq:admm_zeta}
\zeta^{(k+1)}
&=
\argmin_{\zeta\in\mathcal{D}_\delta}\;
\Phi(\zeta)+\frac{\rho}{2}\|\xi^{(k+1)}-\zeta+U^{(k)}\|_F^2,\\
\label{eq:admm_dual}
U^{(k+1)}
&=
U^{(k)}+\xi^{(k+1)}-\zeta^{(k+1)}.
\end{align}
The $\xi$ update reduces to a Sylvester system via the standard diagonalization argument of \citet{zhu2020convex}, while the $\Omega$ block proximal update and the primal-dual stopping rules are summarized in Appendix~\ref{app:gap_framework}. The inner ADMM loop is terminated when
\begin{equation}
\label{eq:pn_gap_rel_main}
\operatorname{gap}_{\mathrm{PN},m}^{\mathrm{rel}}\bigl(\zeta^{(k+1)},W^{(k+1)}\bigr)
\le
\varepsilon_{\mathrm{admm}},
\end{equation}
where the dual certificate $W^{(k+1)}$ and the explicit dual objective are given in Appendix~\ref{app:pn_gap}.

\begin{algorithm}[H]
\caption{ADMM for the PN subproblem \eqref{eq:pn_splitting} at outer iterate $m$}
\label{alg:admm_pn}
\KwIn{$H^{(m)},\Sigma^{(m)},G^{(m)}$; penalty $\rho>0$; inner tolerance $\varepsilon_{\mathrm{admm}}>0$.}
\KwInit{Warm start $\xi^{(0)},\zeta^{(0)},U^{(0)}$.}
\For{$k=0,1,2,\dots$}{
  Update $\xi^{(k+1)}$ by solving \eqref{eq:admm_xi}\;
  Update $\gamma_\zeta^{(k+1)}$ by entrywise soft-thresholding and update $\Omega_\zeta^{(k+1)}$ by Appendix~\ref{app:omega_prox}\;
  Update $U^{(k+1)} \leftarrow U^{(k)}+\xi^{(k+1)}-\zeta^{(k+1)}$\;
  Form the dual feasible certificate $W^{(k+1)}$ and evaluate \eqref{eq:pn_gap_rel_main}\;
  \If{$\operatorname{gap}_{\mathrm{PN},m}^{\mathrm{rel}}\bigl(\zeta^{(k+1)},W^{(k+1)}\bigr)\le \varepsilon_{\mathrm{admm}}$}{
    \textbf{break}\;
  }
}
\KwOut{$s^{(m)} \leftarrow \zeta^{(k+1)}$.}
\end{algorithm}

\begin{remark}[Relation to \citet{zhu2020convex}]
The ADMM splitting above is close to the supplementary algorithm of \citet{zhu2020convex}, but there are two main differences. First, \citet{zhu2020convex} does not establish self-concordance of the reparameterized loss, so the Newton direction in that work is globalized by backtracking. Once Lemma~\ref{lem:self_concordant_Gauss} is available, the PN update uses the self-concordant step
\(
\alpha^{(m)}=(1+\nu_{(m)})^{-1},
\
\nu_{(m)}:=\|d^{(m)}\|_{\xi^{(m)}},
\)
which gives monotone descent, global convergence, and the usual second-order local rate under composite self-concordant theory. Second, the implementation in \citet{zhu2020convex} effectively applies a single soft-thresholding step on the off-diagonal entries of the $\Omega$ block. That shortcut is valid whenever the soft-thresholded candidate already satisfies $\Omega\succeq \delta I_q$, in which case Appendix~\ref{app:omega_prox} also terminates immediately. When the thresholded candidate is not positive definite, however, the PN subproblem needs the additional ADMM correction from Appendix~\ref{app:omega_prox} to enforce the safeguard; without this correction the PN update can fail.
\end{remark}

\subsection{Penalty path construction and warm starts}
\label{sec:reg_path}

The regularization path is initialized at the null model, meaning the largest penalty pair for which the optimizer has no active predictor--response links and a diagonal precision matrix:
\begin{equation}
\label{eq:null_corner_main}
\lambda_\gamma^{\max}
:=
2\norm{\SXY}_{\max},
\quad
\lambda_\Omega^{\max}
:=
\norm{\SYY}_{\max,\mathrm{off}},
\quad
\gamma^{(0)}
=
\mathbf 0,
\quad
\Omega^{(0)}
=
\diag(\SYY)^{-1}.
\end{equation}
The Karush--Kuhn--Tucker (KKT) verification is given in Appendix~\ref{app:proof_null_corner}. More generally, let
\[
\mathcal P=\{(\lambda_\gamma^{(\ell)},\lambda_\Omega^{(\ell)})\}_{\ell=1}^{L}
\]
denote any penalty path chosen by the user. Warm starts are then obtained by initializing each fit from a neighboring previously computed solution. A common choice is the geometric path
\[
\lambda_\gamma^{(\ell)}=\lambda_\gamma^{\max}\eta_\gamma^{\,\ell-1},
\qquad
\lambda_\Omega^{(\ell)}=\lambda_\Omega^{\max}\eta_\Omega^{\,\ell-1},
\qquad
\ell=1,\ldots,L,
\]
where $\eta_\gamma=(r_\gamma)^{1/(L-1)}$ and $\eta_\Omega=(r_\Omega)^{1/(L-1)}$ for chosen ratios $r_\gamma,r_\Omega\in(0,1)$. When a rectangular search over the two penalties is desired, one may instead use the two-dimensional grid
\[
\lambda_{\gamma,a}=\lambda_\gamma^{\max}\eta_\gamma^{\,a-1},\quad a=1,\ldots,L_\gamma,
\qquad
\lambda_{\Omega,b}=\lambda_\Omega^{\max}\eta_\Omega^{\,b-1},\quad b=1,\ldots,L_\Omega,
\]
together with neighborhood warm starts.

\subsection{Raw scale recovery and evaluation metrics}
\label{subsec:raw_eval}

After solving the problem on the working scale for $(\hat\gamma,\hat\Omega)$, the corresponding slope and intercept on the raw scale are recovered by \(\hat\beta^{\mathrm{raw}}=D_x^{-1}\hat\gamma\hat\Omega^{-1}\) and \(\hat\alpha=(\bar y-\bar x\,\hat\beta^{\mathrm{raw}})^\top.\) For the conditional quantities used below it is also convenient to define \(\hat\gamma^{\mathrm{raw}}:=\hat\beta^{\mathrm{raw}}\hat\Omega=D_x^{-1}\hat\gamma\) and \(\hat\gamma_0:=\hat\alpha^\top\hat\Omega=\bar y\,\hat\Omega-\bar x\,D_x^{-1}\hat\gamma.\)

For later simulation and real-data reporting, model selection and out-of-sample assessment are carried out on the raw data scale. For a raw-scale test sample $\{(x_i^{\mathrm{raw}},y_i^{\mathrm{raw}})\}_{i=1}^{n_{\mathrm{te}}}$, we first define the ordinary prediction mean squared error (MSE) \( \operatorname{MSE}_{\mathrm{te}}=\frac{1}{n_{\mathrm{te}}\cdot q}\sum_{i=1}^{n_{\mathrm{te}}}\left\| y_i^{\mathrm{raw}} - \hat\alpha - (\hat\beta^{\mathrm{raw}})^\top x_i^{\mathrm{raw}}\right\|_2^2.\)

To evaluate the conditional structure encoded by $(\hat\gamma,\hat\Omega)$, define for coordinate $j$
\begin{equation*}
\hat m_j^{\mathrm{raw}}(x^{\mathrm{raw}},y_{-j}^{\mathrm{raw}})
=
\frac{\hat\gamma_{0,j} + (x^{\mathrm{raw}})^\top \hat\gamma^{\mathrm{raw}}_{\cdot j} - \hat\Omega_{j,-j}y_{-j}^{\mathrm{raw}}}{\hat\Omega_{jj}},
\end{equation*}
and then the conditional mean squared error (CMSE)
\begin{equation*}
\operatorname{CMSE}_{\mathrm{te}}
=
\frac{1}{n_{\mathrm{te}}\cdot  q}
\sum_{i=1}^{n_{\mathrm{te}}}\sum_{j=1}^{q}
\Bigl\{
 y_{ij}^{\mathrm{raw}} - \hat m_j^{\mathrm{raw}}(x_i^{\mathrm{raw}},y_{i,-j}^{\mathrm{raw}})
\Bigr\}^{2}.
\end{equation*}
For MRCE, the same raw-scale evaluation formulas apply after inserting the fitted raw-scale parameters $(\hat\alpha,\hat\beta,\hat\Omega)$ returned by that method and setting \(\hat\gamma^{\mathrm{raw}}=\hat\beta\hat\Omega\) and \(\hat\gamma_0=\hat\alpha^\top\hat\Omega.\)

\subsection{Convergence guarantees for first- and second-order methods}
\label{subsec:conv_guarantees}
For Algorithm~\ref{alg:proxgrad}, the analytic step size above is the descent rule of \citet[Lemma~12]{tran2015composite}, so each accepted proximal-gradient iterate decreases the objective. If the relevant level set is bounded and the diagonal metric is uniformly bounded below, the iterates converge globally to the unique minimizer by \citet[Theorem~13]{tran2015composite}; under additional strong regularity and Hessian approximation conditions, they converge locally at a linear rate by \citet[Theorem~14]{tran2015composite}. For Algorithm~\ref{alg:proxnewton}, the damping $\alpha^{(m)}=(1+\nu_{(m)})^{-1}$ is the optimal self-concordant step in \citet[Theorem~6]{tran2015composite}, guaranteeing monotone decrease. Once the Newton decrement is small enough, \citet[Theorem~7]{tran2015composite} gives local quadratic convergence, matching the classical self-concordant Newton behavior in \citet{nesterov2004introductory}.

\section{Simulation Study}\label{sec:sim}

We report a simulation study focused on algorithmic robustness and the pathwise computational behavior of the proposed \textsc{SCA-MRCE} solvers. We generate a prespecified sparse multivariate linear regression model with Gaussian errors, fit the full regularization path on a single realization, and compare the terminated penalized objective values, elapsed time, and outer iteration counts across solvers. This design lets us compare the solvers on a common problem instance without additional Monte Carlo variation.

Let $X_i^{\mathrm{raw}}\in\R^p$ and $Y_i^{\mathrm{raw}}\in\R^q$ follow the independently and identically distributed (i.i.d.) model \(X_i^{\mathrm{raw}}
\stackrel{\mathrm{i.i.d.}}{\sim}
N_p(0,I_p),
\
Y_i^{\mathrm{raw}}=
(\beta^\star)^\top X_i^{\mathrm{raw}}+\varepsilon_i,
\
\varepsilon_i
\stackrel{\mathrm{i.i.d.}}{\sim}
N_q\bigl(0,(\Omega^\star)^{-1}\bigr).\)
Because the penalized convex program is formulated in $(\gamma,\Omega)$, sparsity is imposed on the natural parameter \(\gamma^\star=\beta^\star\Omega^\star,\) rather than directly on $\beta^\star$. Specifically, the support of $\gamma^\star$ is generated at a prespecified sparsity level (5\% nonzeros per response) on the $p\times q$ coefficient array; conditional on that support, the nonzero entries are assigned independent random signs and magnitudes drawn from a distribution bounded away from zero, and all remaining entries are set to zero. We generate the precision matrix $\Omega$ by independently activating each upper triangular off-diagonal entry with probability $0.05$, assigning each active edge weight $0.5$, symmetrizing the matrix, and adding a diagonal shift so that the resulting sparse positive-definite precision matrix has the prescribed condition number \(\kappa(\Omega^\star)\in\{2,10,100\}\). Finally, we set \(\beta^\star=\gamma^\star(\Omega^\star)^{-1}.\)
This is the appropriate simulation target because the penalized convex program acts on $(\gamma,\Omega)$.

We use a single realized data set with $n=1000$, $p=q=50$, and a geometric path of length $60$ initialized at the closed-form null model from Section~\ref{sec:reg_path}, with ratio $10^{-4}$ between the smallest and largest penalties. We consider three conditioning regimes for the precision matrix, $\kappa(\Omega^\star)\in\{2,10,100\}$. Warm starts are used along the path. We compare four \textsc{SCA-MRCE} solvers: \texttt{Legacy\_PN}, \texttt{Aggressive\_PN}, \texttt{Legacy\_PG}, and \texttt{Aggressive\_PG}. The legacy variants are conservative reference implementations with larger inner solve budgets. The aggressive variants truncate the same outer algorithms: \texttt{Aggressive\_PN} caps the PN subproblem ADMM at 100 iterations and the $\Omega$ proximal ADMM at one iteration per outer step, while \texttt{Aggressive\_PG} truncates the $\Omega$ proximal update to one iteration along a PG path. The comparison therefore isolates whether reducing the expensive inner work changes the pathwise solutions materially or mainly lowers the computational cost.

\begin{figure}[htbp]
    \centering
    \begin{minipage}[b]{0.315\textwidth}
        \centering
        \includegraphics[width=\textwidth]{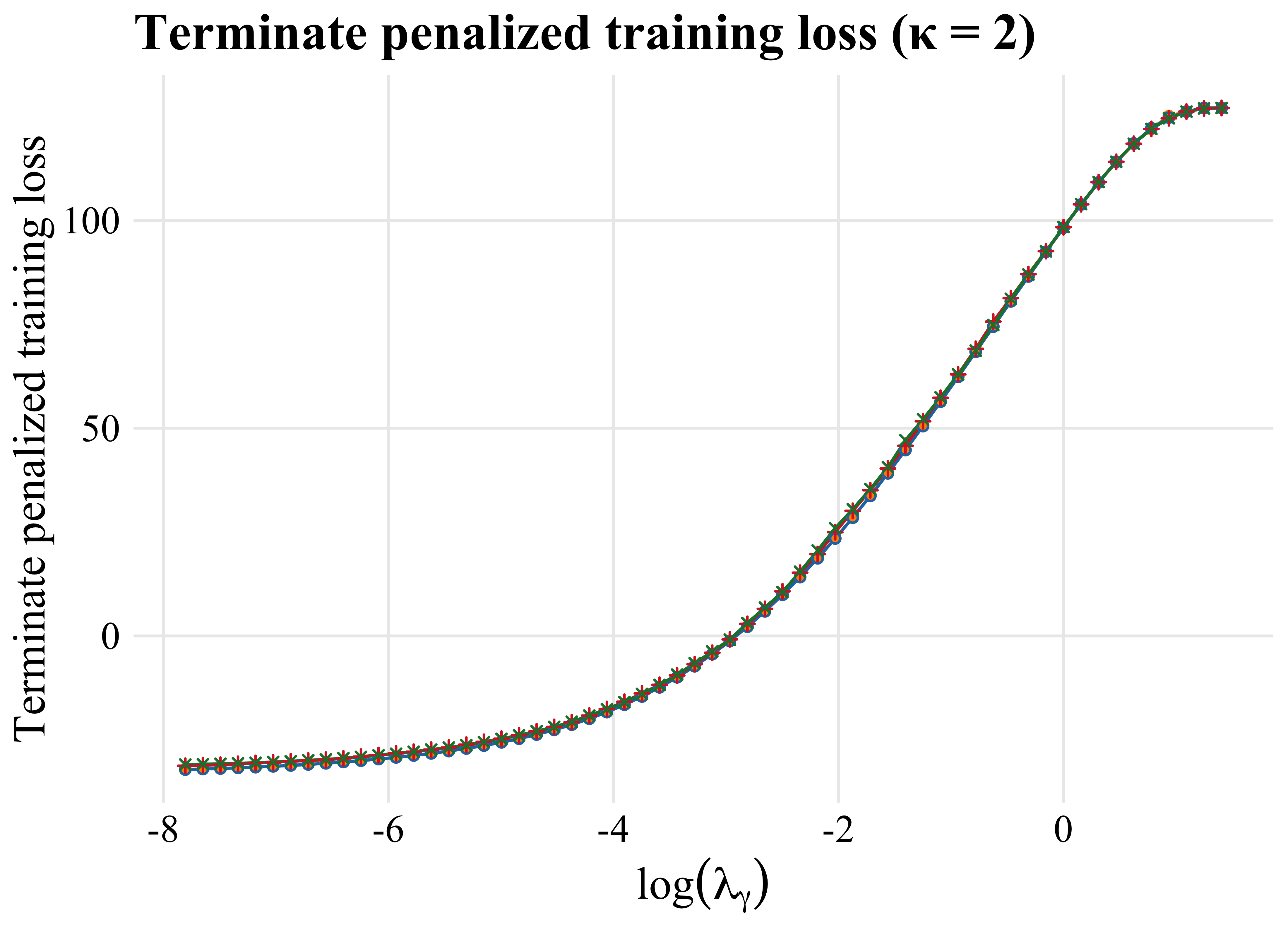}
    \end{minipage}
    \hfill
    \begin{minipage}[b]{0.315\textwidth}
        \centering
        \includegraphics[width=\textwidth]{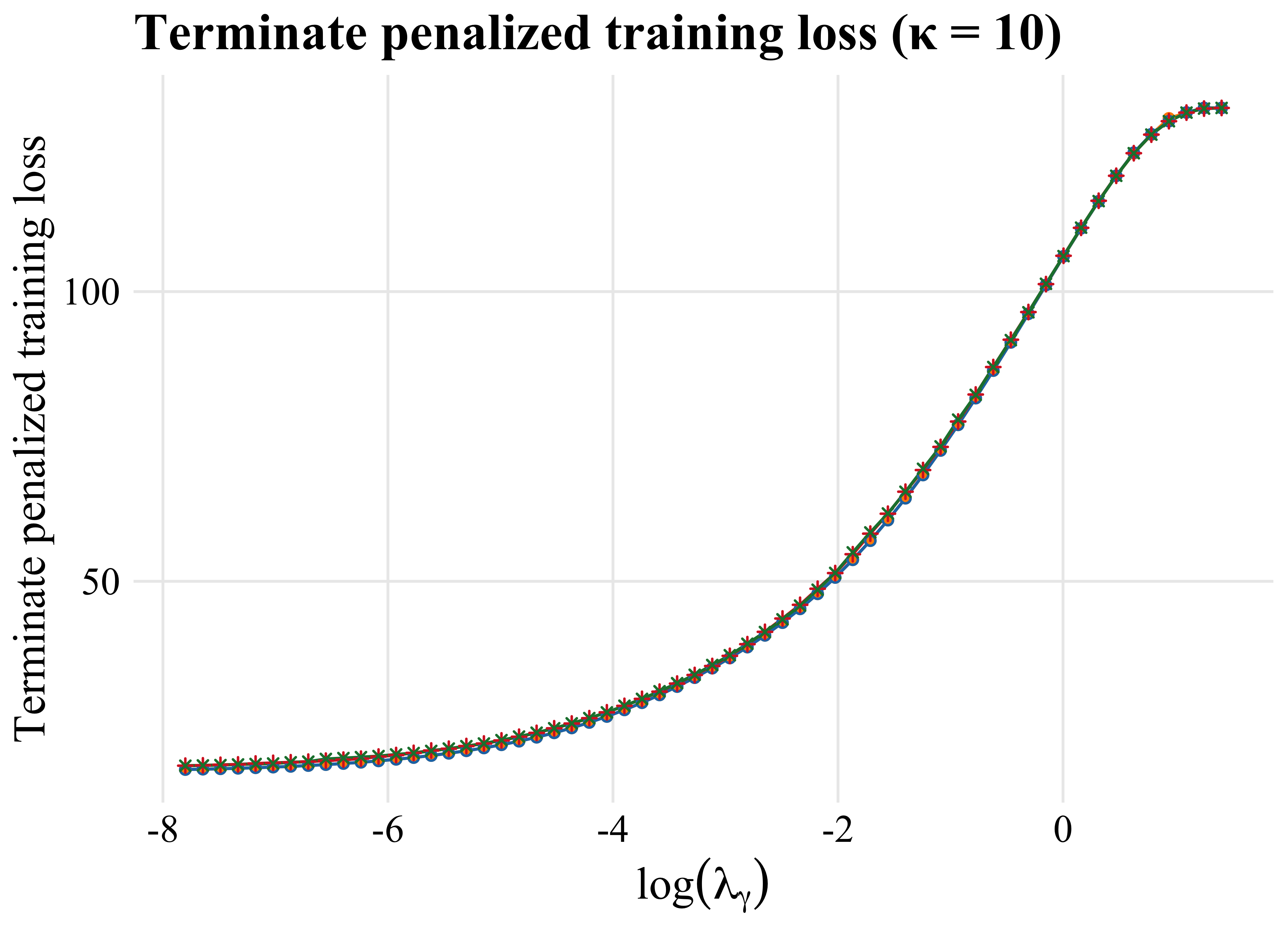}
    \end{minipage}
    \hfill
    \begin{minipage}[b]{0.315\textwidth}
        \centering
        \includegraphics[width=\textwidth]{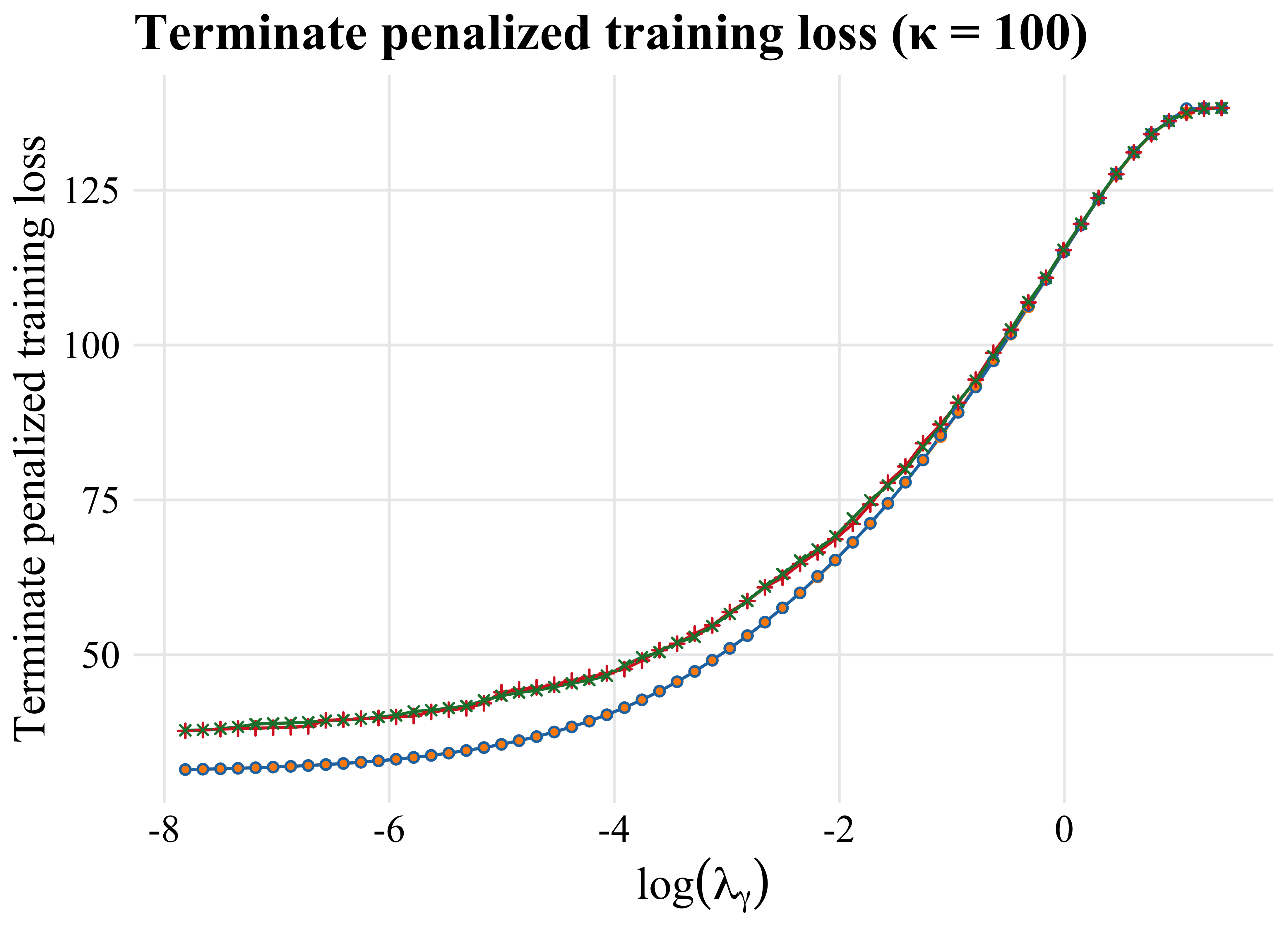}
    \end{minipage}

    \vspace{0.5em}

    \begin{minipage}[b]{0.315\textwidth}
        \centering
        \includegraphics[width=\textwidth]{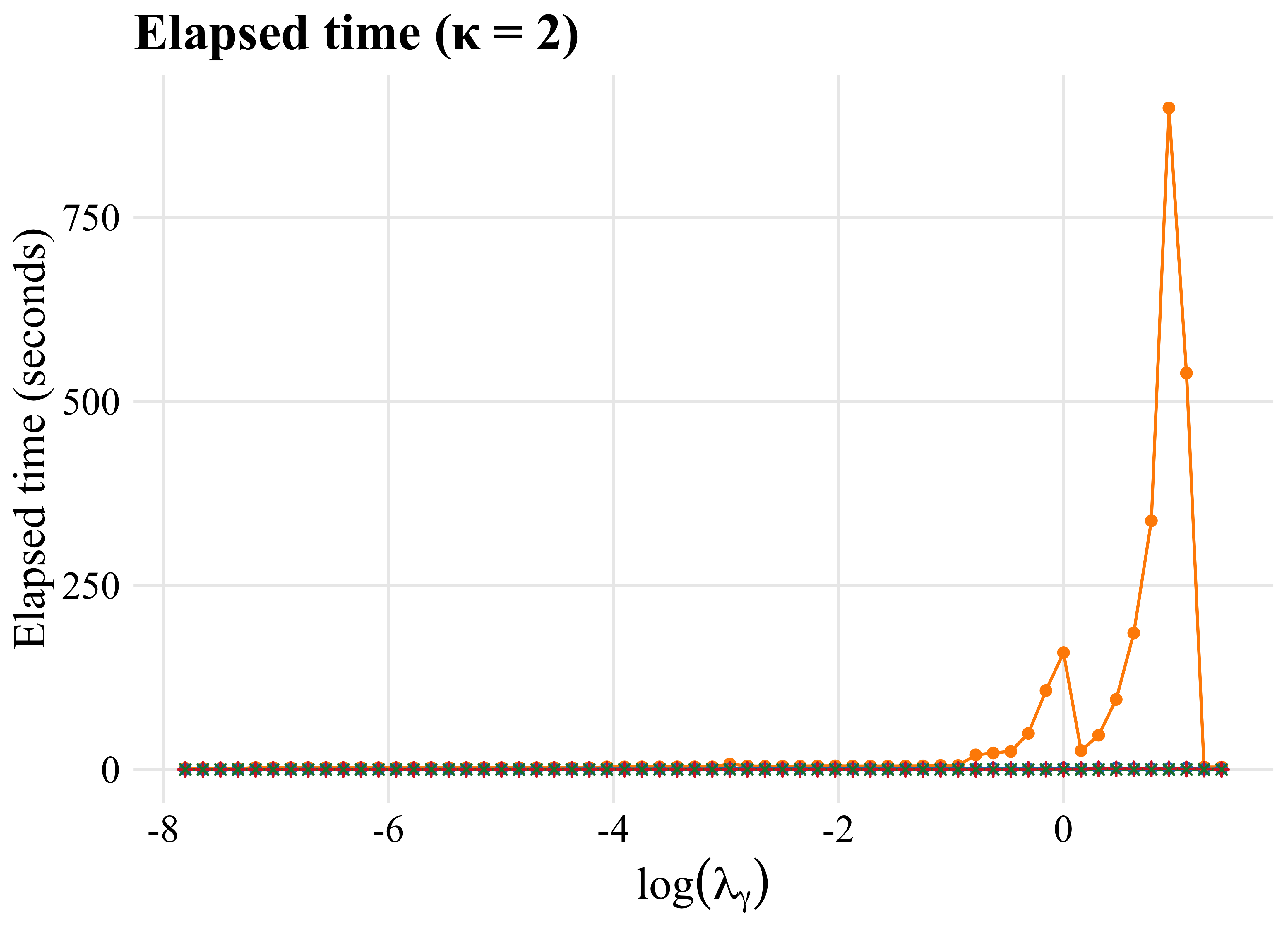}
    \end{minipage}
    \hfill
    \begin{minipage}[b]{0.315\textwidth}
        \centering
        \includegraphics[width=\textwidth]{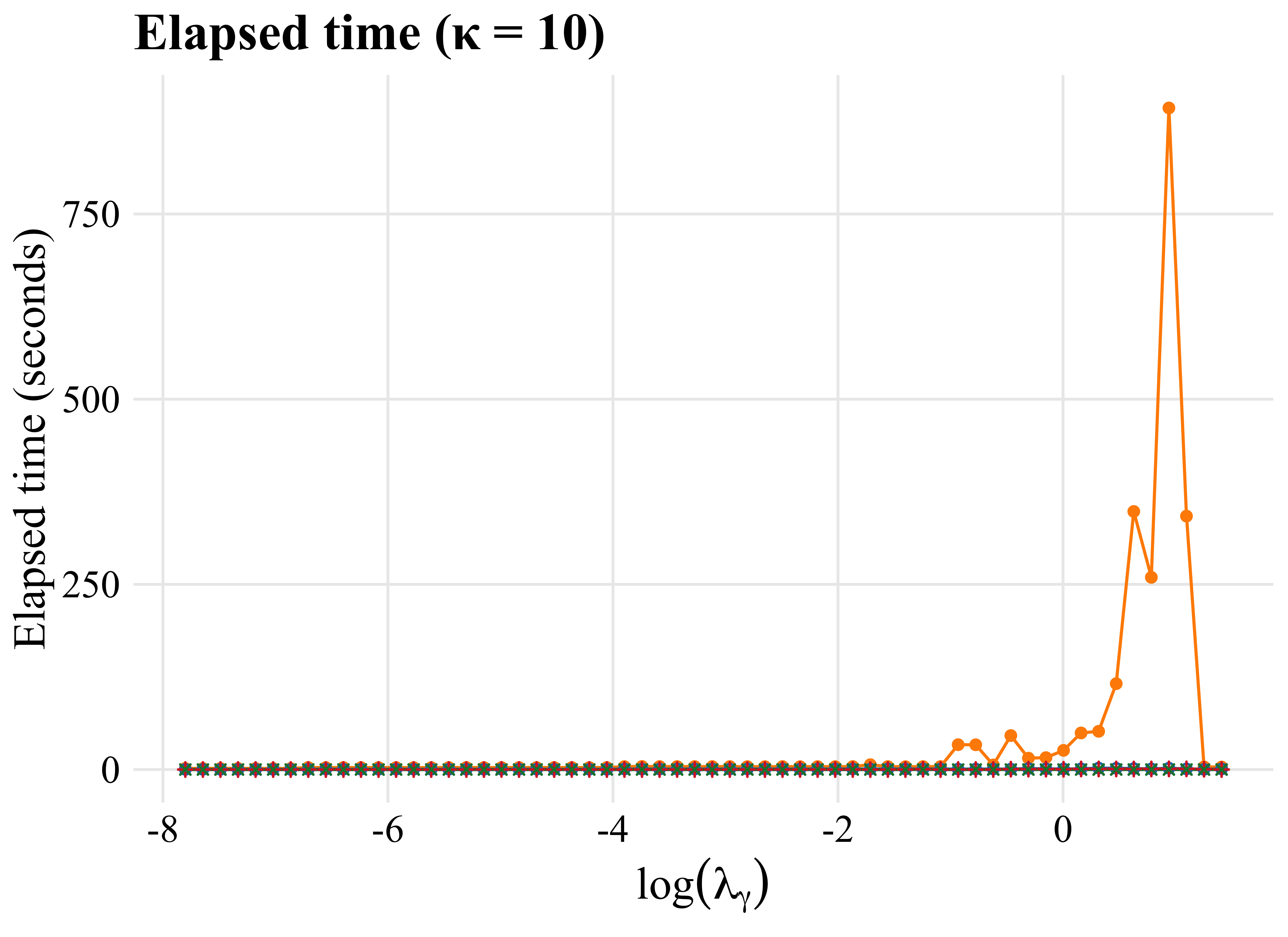}
    \end{minipage}
    \hfill
    \begin{minipage}[b]{0.315\textwidth}
        \centering
        \includegraphics[width=\textwidth]{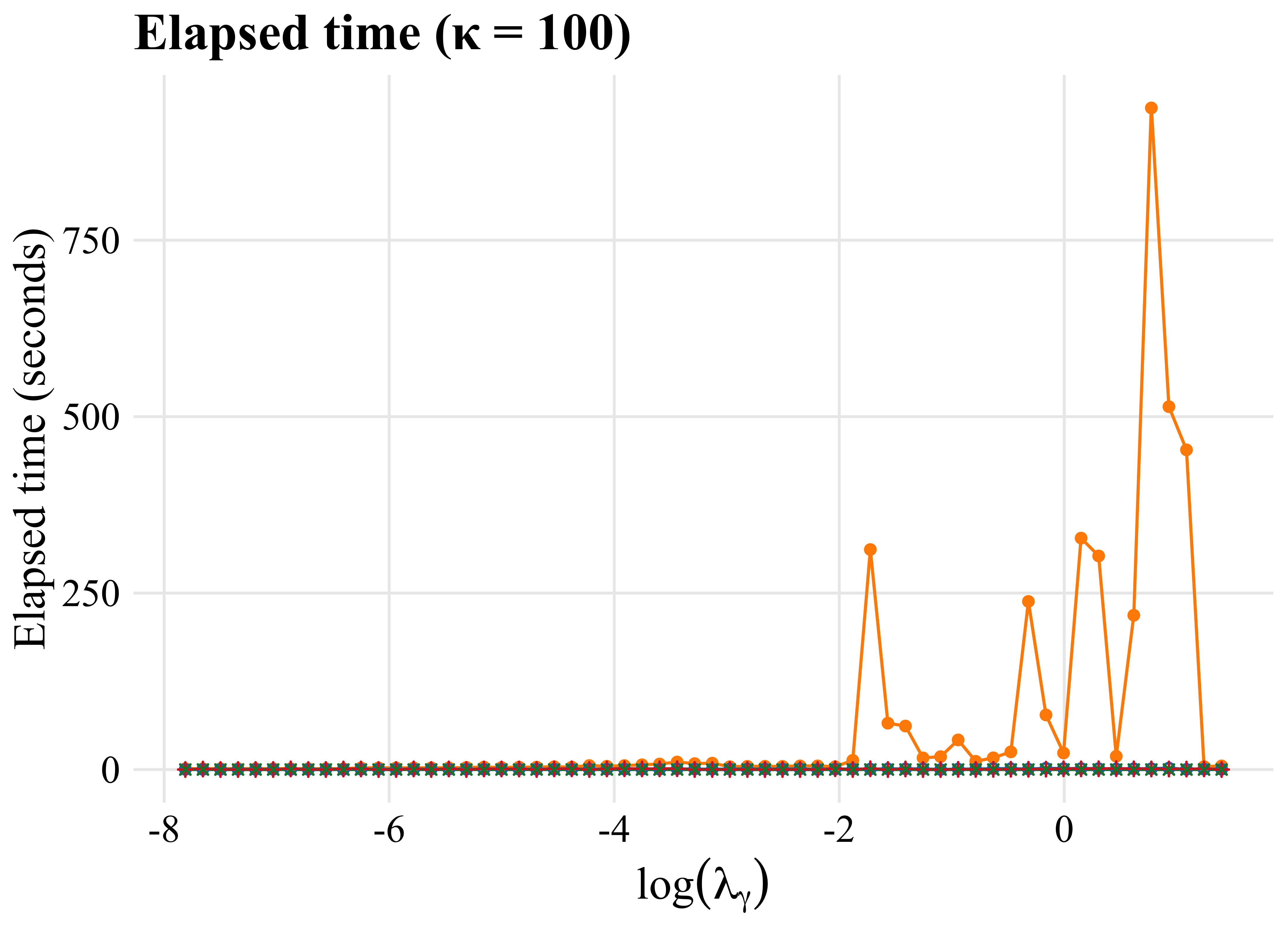}
    \end{minipage}

    \vspace{0.5em}

    \begin{minipage}[b]{0.315\textwidth}
        \centering
        \includegraphics[width=\textwidth]{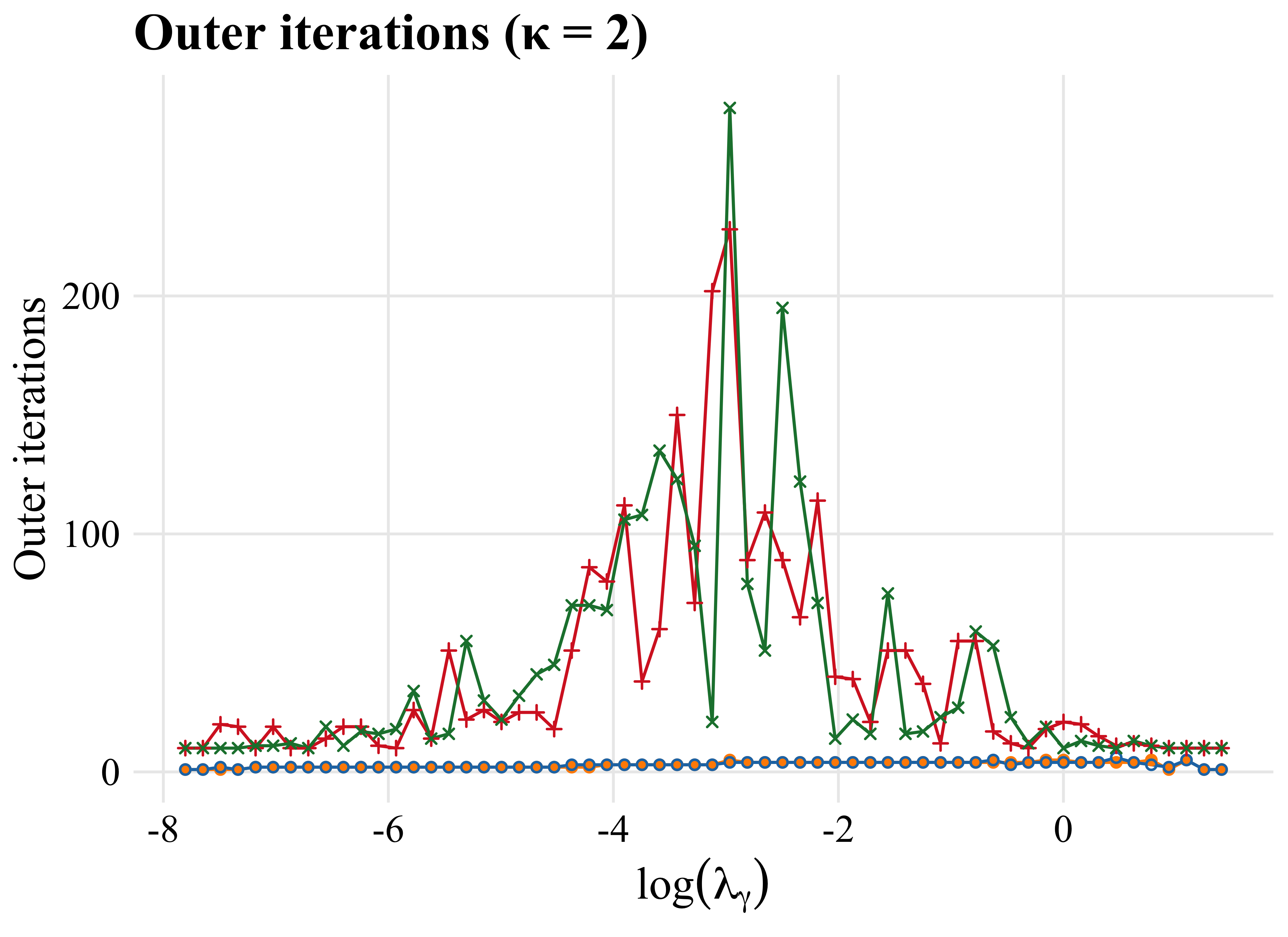}
    \end{minipage}
    \hfill
    \begin{minipage}[b]{0.315\textwidth}
        \centering
        \includegraphics[width=\textwidth]{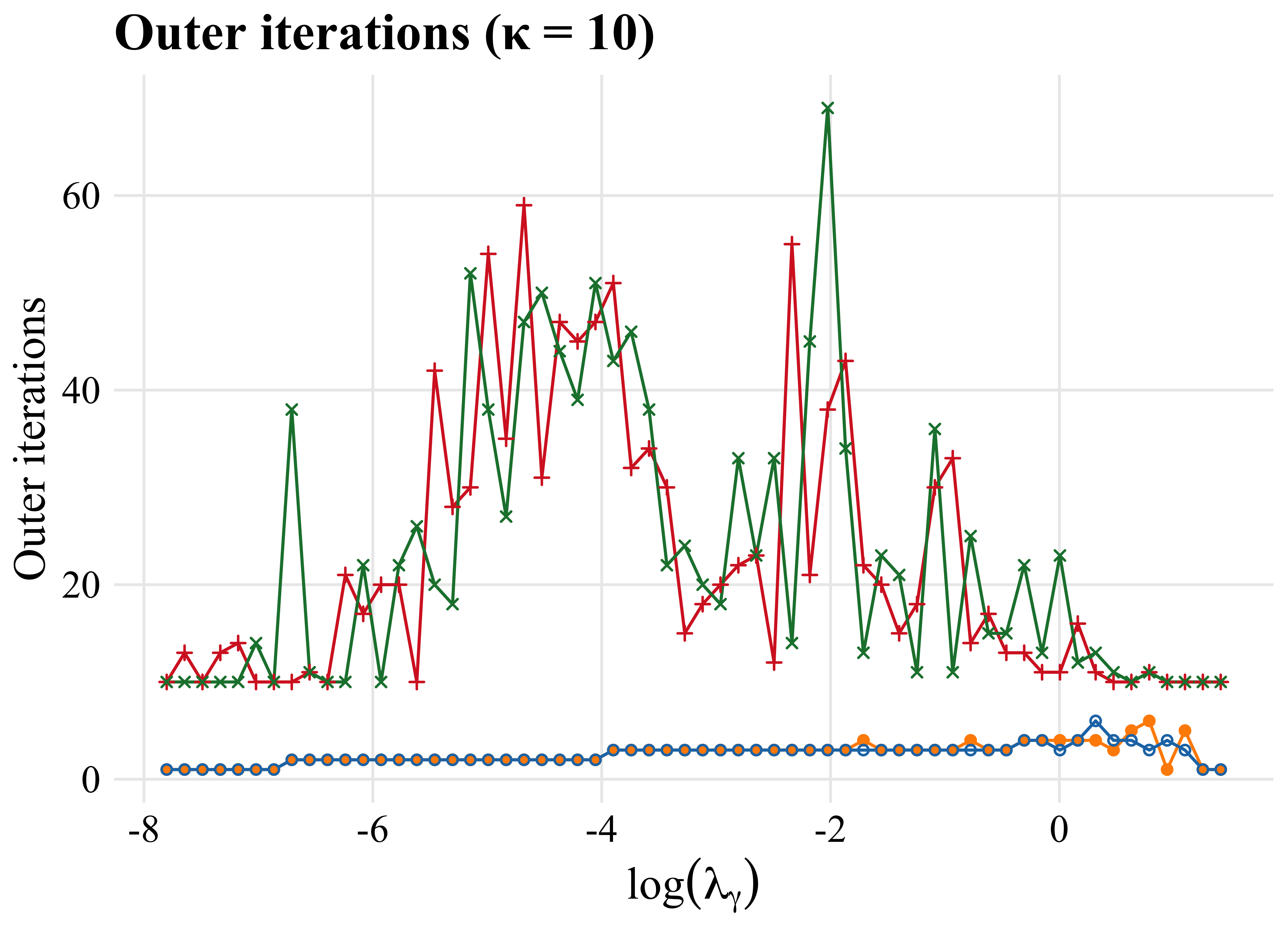}
    \end{minipage}
    \hfill
    \begin{minipage}[b]{0.315\textwidth}
        \centering
        \includegraphics[width=\textwidth]{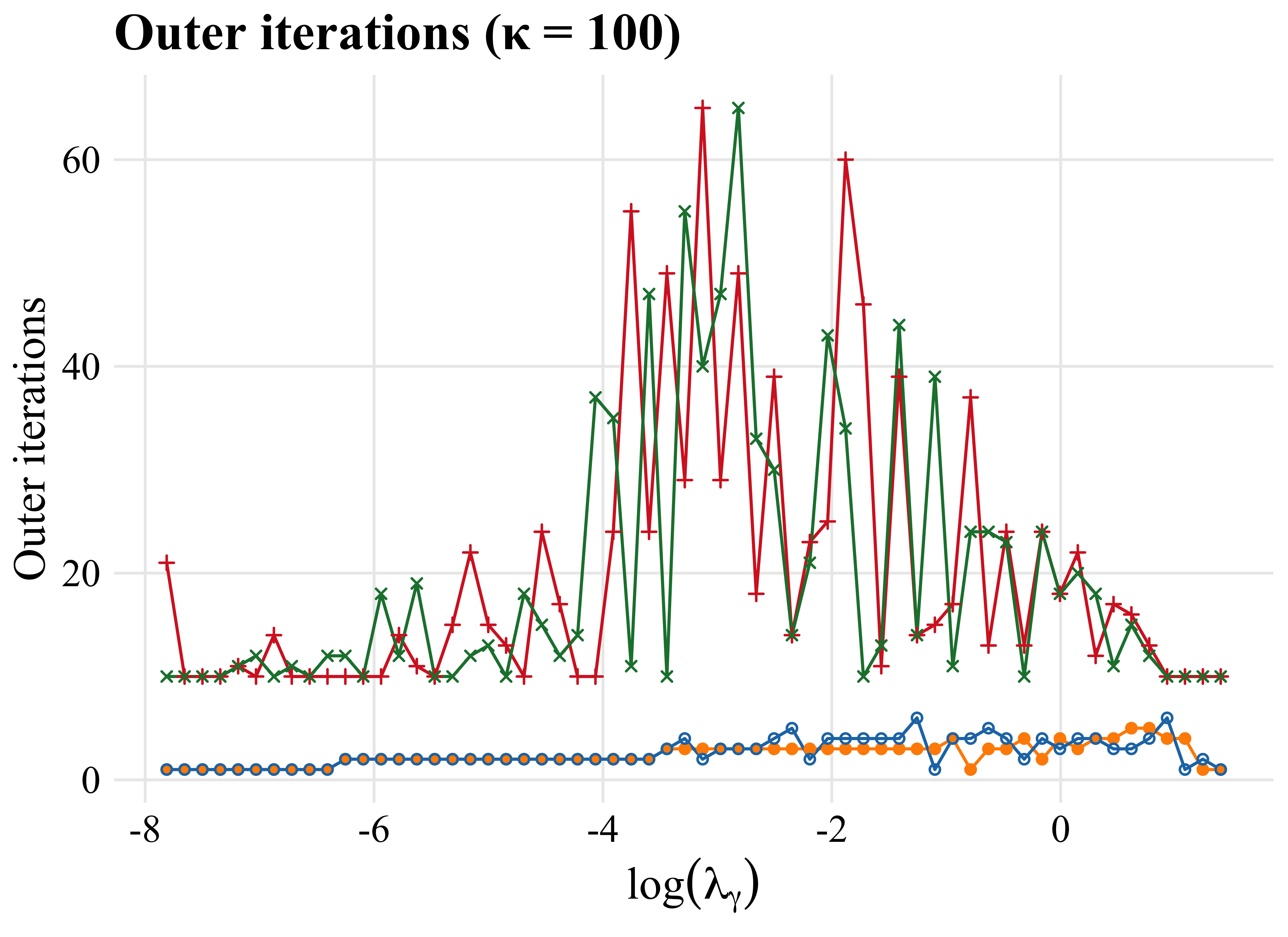}
    \end{minipage}

    \caption{Pathwise computational comparison for the four \textsc{SCA-MRCE} solver configurations. Columns correspond to $\kappa(\Omega^\star)\in\{2,10,100\}$ from left to right. The top row reports the terminated penalized training loss, the middle row reports elapsed time in seconds, and the bottom row reports the number of outer iterations to termination. In all panels the horizontal axis is $\log(\lambda_{\gamma})$. Orange filled circles denote \texttt{Legacy\_PN}, blue open circles denote \texttt{Aggressive\_PN}, red plus markers denote \texttt{Legacy\_PG}, and green x markers denote \texttt{Aggressive\_PG}. The aggressive variants track the legacy objective values closely, while the main computational differences appear in the amount of inner work required at each path point.}
    \label{fig:solver_benchmark_grid}
\end{figure}

\begin{table}[htbp]
\centering
\small
\begin{tabular}{lccc}
\toprule
Solver & Mean terminated penalized loss & Total elapsed time (sec) & Mean outer iterations \\
\midrule
\texttt{Legacy\_PN} & 45.6635 & 8893.73 & 2.66 \\
\texttt{Aggressive\_PN} & 45.6688 & 113.67 & 2.71 \\
\texttt{Legacy\_PG} & 47.5741 & 46.56 & 28.10 \\
\texttt{Aggressive\_PG} & 47.6367 & 4.88 & 28.39 \\
\bottomrule
\end{tabular}
\caption{Summary of the fixed-data pathwise comparison. Each row aggregates results over three condition-number settings, $\kappa(\Omega^\star)\in\{2,10,100\}$, and $60$ penalty pairs per setting. ``Mean terminated penalized loss'' is the average final penalized training objective at solver termination over that $3\times 60$ grid. ``Total elapsed time (sec)'' is the cumulative elapsed time over all runs for the corresponding solver, and ``Mean outer iterations'' is the average number of outer-loop updates per penalty pair.}
\label{tab:sim_solver_benchmark}
\end{table}

Figure~\ref{fig:solver_benchmark_grid} and Table~\ref{tab:sim_solver_benchmark} point to the same conclusion. PN uses far fewer outer iterations than PG on the same path, but its elapsed time depends strongly on how accurately the inner second-order subproblems are solved. Within the PG family, \texttt{Legacy\_PG} and \texttt{Aggressive\_PG} produce very similar terminated losses and outer-iteration profiles, and \texttt{Aggressive\_PG} is the fastest implementation overall; when the soft-thresholded $\Omega$ candidate already satisfies the safeguard, the two PG updates are nearly identical. PN remains attractive when a more accurate solution is needed, but for routine path computation the PG variants are usually the better choice in elapsed time. Truncating the PN inner solves reduces cost substantially, although the resulting inexact steps are more sensitive than the PG updates and can stall at difficult path points.

\section{Real Data Analysis}
\label{sec:realdata}

We apply classical-parameterization MRCE and two \textsc{SCA-MRCE} implementations, \texttt{Legacy\_PG} and \texttt{Aggressive\_PG}, to the Mice Protein Expression data set \citep{higuera2015mice}. The data contain measurements on $72$ mice, each observed at $15$ time points, for a total of $1080$ samples. We take the $q=77$ protein measurements as the multivariate response $Y^{\mathrm{raw}}$, and we construct $X^{\mathrm{raw}}$ from the experimental factors Genotype, Treatment, and Behavior. Across $100$ replicates, every mouse contributes exactly $5$ rows to training, $5$ to validation, and $5$ to test, so $n_{\mathrm{tr}}=n_{\mathrm{val}}=n_{\mathrm{te}}=360$. Predictor preprocessing follows Section~\ref{subsec:centering_intercept}: training predictors are standardized, training responses are centered only, and validation and test evaluations are carried out on the raw scale. Each method is tuned over the same $7\times 7$ geometric grid with ratio $10^{-5}$ between the smallest and largest penalties on both axes, and the tuning rule is validation MSE. Final reporting uses the test MSE and test CMSE on the raw scale defined in Section~\ref{subsec:raw_eval}.

Table~\ref{tab:realdata_summary} summarizes the study over $100$ replicates. Classical-parameterization MRCE, \texttt{Legacy\_PG}, and \texttt{Aggressive\_PG} have very similar validation and test MSE, with differences appearing only in the fourth decimal place. The two \textsc{SCA-MRCE} implementations also have smaller and less variable test CMSE than MRCE. The two PG variants are essentially indistinguishable in the reported prediction metrics, which is consistent with the fact that at most grid points the soft-thresholded $\Omega$ update already satisfies the safeguard and therefore makes the legacy and aggressive PG steps coincide. The main difference is computational: the recorded average fitting time is $2661.271$ seconds for MRCE and about $6$ seconds for the two PG variants, a reduction of roughly $4.5\times 10^2$ in average elapsed time per replicate.

\begin{table}[htbp]
\centering
\small
\begin{tabular}{lcccc}
\toprule
Method & Validation MSE & Test MSE & Test CMSE & Elapsed time (sec) \\
\midrule
MRCE & 0.057189 (0.001722) & 0.057271 (0.001710) & 0.025568 (0.018721) & 2661.271 \\
\texttt{Legacy\_PG} & 0.057227 (0.001701) & 0.057354 (0.001657) & 0.020665 (0.001225) & 5.916 \\
\texttt{Aggressive\_PG} & 0.057227 (0.001701) & 0.057354 (0.001657) & 0.020665 (0.001225) & 5.987 \\
\bottomrule
\end{tabular}
\caption{Summary of the real data study across $100$ replicates. The first three metric columns report mean (standard deviation) across the repeated train/validation/test splits. The elapsed time column reports the corresponding average fitting time recorded by the implementation.}
\label{tab:realdata_summary}
\end{table}

Figures~\ref{fig:realdata_validation} and \ref{fig:realdata_time} display heatmaps for one replicate; they are visual summaries, not averages over $100$ replicates. In every panel, the horizontal axis is $\lambda_{\Omega}$; the vertical axis is $\lambda_{\beta}$ for MRCE and $\lambda_{\gamma}$ for the two \textsc{SCA-MRCE} methods. Figure~\ref{fig:realdata_validation} is consistent with Table~\ref{tab:realdata_summary}: there is no evidence that the convex reparameterization sacrifices predictive accuracy. The selected models (red diamonds) lie away from the most heavily regularized row, and the two PG heatmaps are nearly indistinguishable.

\begin{figure}[htbp]
\centering
\includegraphics[width=\textwidth]{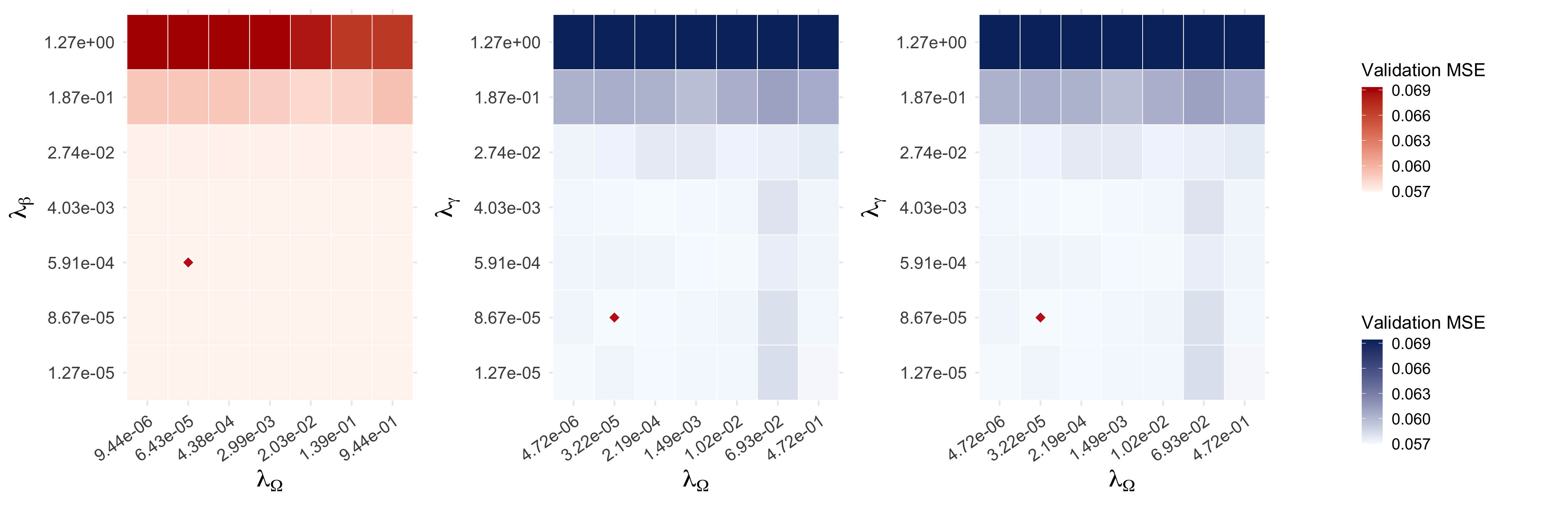}
\caption{Representative validation MSE heatmaps from one replicate. From left to right the panels correspond to MRCE, \texttt{Legacy\_PG}, and \texttt{Aggressive\_PG}. The legend column on the right stacks the MRCE red/white legend above the shared PG blue/white legend. The red diamonds identify the models selected by validation MSE. The horizontal axis is $\lambda_{\Omega}$ throughout; the vertical axis is $\lambda_{\beta}$ for MRCE and $\lambda_{\gamma}$ for the two \textsc{SCA-MRCE} methods.}
\label{fig:realdata_validation}
\end{figure}

Figure~\ref{fig:realdata_time} shows the computational difference more clearly. The MRCE panel exhibits substantial inflation in elapsed time as the penalties move into the weakly regularized region, where the fitted precision matrix can become dense, whereas both PG implementations remain below one second per grid point in this representative replicate. A plausible explanation is that MRCE solves a nonconvex alternating problem and that the corresponding \texttt{rblasso}/\texttt{blasso} computations become more expensive when the fitted $\Omega$ is dense. By contrast, the convex PG formulations remain numerically stable across the grid, and their elapsed time varies much less.

\begin{figure}[htbp]
\centering
\includegraphics[width=\textwidth]{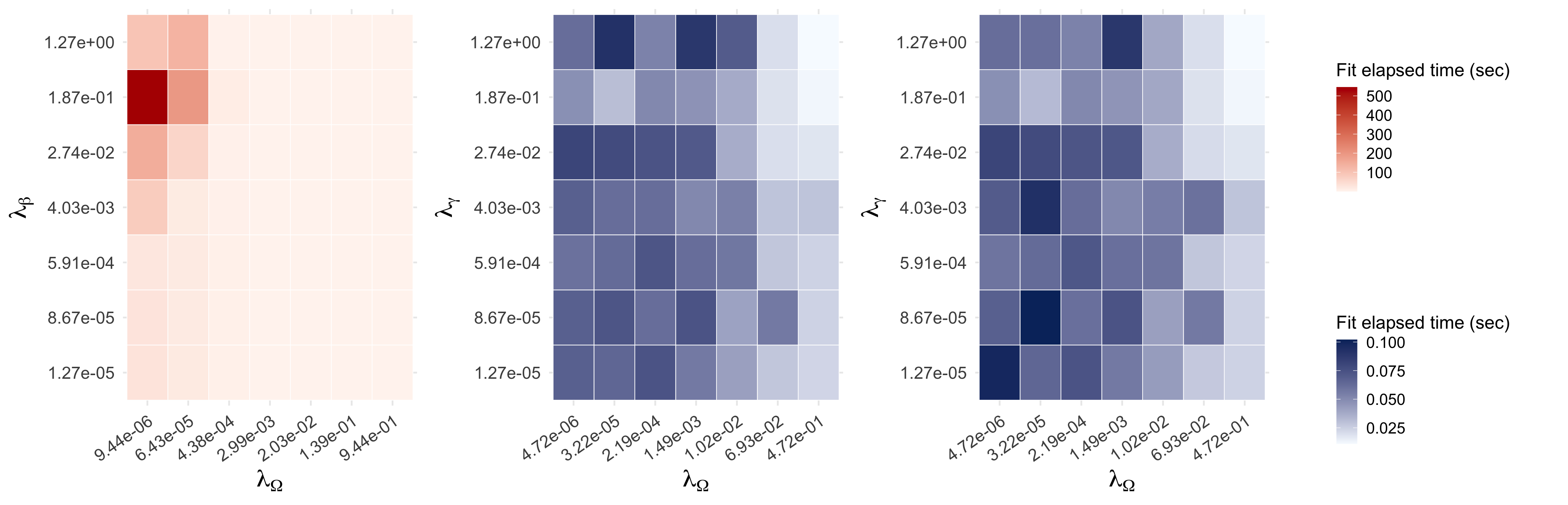}
\caption{Representative elapsed time heatmaps from one replicate. From left to right the panels correspond to MRCE, \texttt{Legacy\_PG}, and \texttt{Aggressive\_PG}. The legend column on the right stacks the MRCE red/white elapsed time legend above the shared PG blue/white legend. The horizontal axis is $\lambda_{\Omega}$ throughout; the vertical axis is $\lambda_{\beta}$ for MRCE and $\lambda_{\gamma}$ for the two \textsc{SCA-MRCE} methods.}
\label{fig:realdata_time}
\end{figure}

The real data analysis is consistent with the simulation study. On this protein expression problem, the convex \textsc{SCA-MRCE} formulation matches classical-parameterization MRCE in ordinary prediction error, improves the conditional prediction metric, and greatly reduces the cost of repeated grid fitting.

\section{Conclusion}
\label{sec:conclusion}

We studied a convex formulation of multivariate linear regression with Gaussian errors based on the reparameterization $(\gamma,\Omega)$ and showed that the corresponding Gaussian loss is standard self-concordant. This result places the penalized likelihood problem within the composite self-concordant framework and justifies both proximal gradient and damped proximal Newton algorithms.

Numerically, proximal Newton uses many fewer outer iterations, but its overall cost depends on the inner solves. For penalty path computation, the proximal gradient variants were more stable and usually faster in elapsed time. The aggressive variants were often faster still, but they do not have the same theoretical convergence guarantees. In the protein expression example, the convex formulation matched classical-parameterization MRCE in ordinary prediction error, improved the conditional prediction metric, and reduced computation time substantially.

\appendix

\section{Proof of Lemma \ref{lem:centering_intercept}}
\label{app:proof_centering}

\begin{proof}
Fix $(\beta^{\mathrm{raw}},\Omega)$ with $\Omega\succ0$, and write
\begin{align*}
r_i(\alpha,\beta^{\mathrm{raw}})
&:=
y_i^{\mathrm{raw}}-\alpha-(\beta^{\mathrm{raw}})^\top x_i^{\mathrm{raw}}, \\
\LG_{\mathrm{raw}}(\alpha,\beta^{\mathrm{raw}},\Omega)
&=
-\log\det(\Omega)
+\frac1n\sum_{i=1}^n r_i(\alpha,\beta^{\mathrm{raw}})^\top \Omega\, r_i(\alpha,\beta^{\mathrm{raw}}).
\end{align*}
Since $\Omega\succ0$, the map $\alpha\mapsto \LG_{\mathrm{raw}}(\alpha,\beta^{\mathrm{raw}},\Omega)$ is strictly convex, and
\[
\nabla_\alpha \LG_{\mathrm{raw}}(\alpha,\beta^{\mathrm{raw}},\Omega)
=
-\frac{2}{n}\sum_{i=1}^n \Omega\bigl(y_i^{\mathrm{raw}}-\alpha-(\beta^{\mathrm{raw}})^\top x_i^{\mathrm{raw}}\bigr)
=
-2\,\Omega\bigl(({\bar y-\bar x\,\beta^{\mathrm{raw}}})^\top-\alpha\bigr).
\]
Therefore the unique minimizer in $\alpha$ is \(\alpha^\star(\beta^{\mathrm{raw}})=(\bar y-\bar x\, \beta^{\mathrm{raw}})^\top.\)

Now set $\beta=D_x\beta^{\mathrm{raw}}$, so that $\beta^{\mathrm{raw}}=D_x^{-1}\beta$. Substituting $\alpha^\star(\beta^{\mathrm{raw}})$ back into the residuals gives
\begin{equation*}
y_i^{\mathrm{raw}}-\alpha^\star(\beta^{\mathrm{raw}})-(\beta^{\mathrm{raw}})^\top x_i^{\mathrm{raw}}
=
\bigl(y_i^{\mathrm{raw}}-\bar y^\top\bigr)-\beta^\top D_x^{-1}\bigl(x_i^{\mathrm{raw}}-\bar x^\top\bigr)
=
y_i-\beta^\top x_i.
\end{equation*}
Hence
\[
\LG_{\mathrm{raw}}\bigl(\alpha^\star(\beta^{\mathrm{raw}}),\beta^{\mathrm{raw}},\Omega\bigr)
=
-\log\det(\Omega)
+\frac1n\sum_{i=1}^n
\bigl(y_i-\beta^\top x_i\bigr)^\top
\Omega
\bigl(y_i-\beta^\top x_i\bigr)
=
\LG_{\mathrm{prep}}(\beta,\Omega).
\]
Any penalty that does not involve $\alpha$ is unchanged by this substitution, so the penalized problem on the raw data with an unpenalized intercept is equivalent to the corresponding preprocessed problem without an explicit intercept.

\end{proof}

\section{Proof of Lemma \ref{lem:self_concordant_Gauss}}
\label{app:proof_self_concordant}

\begin{proof}

Fix $\xi=(\gamma,\Omega)\in\dom(\LG)$ and a direction $d\xi=(d\gamma,d\Omega)\in\Xi$.
Let
\(  
\phi(t):=\LG(\xi+t\,d\xi),
\)
for all $t$ such that $\Omega+t\,d\Omega\succ0$.
By Definition~\ref{def:self_concordant_1d}, it suffices to show
\( 
|\phi'''(0)|\le 2\{\phi''(0)\}^{3/2}.
\)

Write
\( 
\phi(t)= -\log\det(\Omega+t\,d\Omega) + g(t),
\)
where
\[
g(t)
:=
\tr\Bigl(
  (\Omega+t\,d\Omega)^{-1}(\gamma+t\,d\gamma)^\top \SXX(\gamma+t\,d\gamma)
  -2\,\SXY^\top(\gamma+t\,d\gamma)
  +(\Omega+t\,d\Omega)\SYY
\Bigr).
\]
The last two terms in $g(t)$ are affine in $(\gamma,\Omega)$, hence they contribute
neither to $\phi''(0)$ nor to $\phi'''(0)$.

Introduce the normalized quantities
\( 
  \bar d\Omega := \Omega^{-1/2} d\Omega\,\Omega^{-1/2}\in\Symq,\ 
  \bar\gamma := \SXX^{1/2}\gamma\,\Omega^{-1/2},\ 
  \bar d\gamma := \SXX^{1/2}d\gamma\,\Omega^{-1/2},
\)
and define $\Delta:=\bar d\gamma-\bar\gamma\,\bar d\Omega$.
Using
\((\Omega+t\,d\Omega)^{-1}
=
\Omega^{-1/2}(I+t\bar d\Omega)^{-1}\Omega^{-1/2},\) the only part of $g(t)$ that is not affine becomes 
\(\tr\Bigl((I+t\bar d\Omega)^{-1}(\bar\gamma+t\bar d\gamma)^\top(\bar\gamma+t\bar d\gamma)\Bigr).\)
Now set \(A:=\bar d\Omega, B:=\bar\gamma, C:=\bar d\gamma.\)
Since \((I+tA)^{-1}=I-tA+t^2A^2-t^3A^3+o(t^3),\)
a direct multiplication gives
\begin{align*}
\tr\Bigl((I+tA)^{-1}(B+tC)^\top(B+tC)\Bigr)
&=
c_0+t c_1+t^2\tr\bigl((C-BA)^\top(C-BA)\bigr) \\
&\quad
-t^3\tr\bigl((C-BA)^\top(C-BA)A\bigr)+o(t^3).
\end{align*}
for some scalars $c_0,c_1$. Therefore,
\(g''(0)=2\|\Delta\|_F^2,
g'''(0)=-6\,\tr(\Delta^\top\Delta\,\bar d\Omega).\)

For the log determinant term,
\( -\log\det(\Omega+t\,d\Omega)
=
-\log\det(\Omega)-\log\det(I+t\bar d\Omega),\)
and the Taylor expansion of $\log\det(I+t\bar d\Omega)$ yields
\(
\frac{d^2}{dt^2}\Bigl[-\log\det(\Omega+t\,d\Omega)\Bigr]_{t=0}
=
\|\bar d\Omega\|_F^2\), and
\(
\frac{d^3}{dt^3}\Bigl[-\log\det(\Omega+t\,d\Omega)\Bigr]_{t=0}
=
-2\,\tr\bigl((\bar d\Omega)^3\bigr).
\)
Hence
\(
\phi''(0)=
\|\bar d\Omega\|_F^2 + 2\|\Delta\|_F^2,
\phi'''(0)=
-2\,\tr\bigl((\bar d\Omega)^3\bigr)
-6\,\tr\bigl(\Delta^\top\Delta\,\bar d\Omega\bigr).
\)
In particular, $\phi''(0)\ge 0$ for every feasible $\xi$ and $d\xi$. Since the same
calculation applies after shifting the base point along any admissible line, every
one-dimensional restriction of $\LG$ has nonnegative second derivative; therefore
$\LG$ is convex.

Consequently,
\(
|\phi'''(0)|
\le
2\|\bar d\Omega\|_F^3 + 6\|\Delta\|_F^2\,\|\bar d\Omega\|_F,
\)
where we used
\(
|\tr((\bar d\Omega)^3)|\le \|\bar d\Omega\|_F^3
\)
and
\(
|\tr(\Delta^\top\Delta\,\bar d\Omega)|
\le \|\Delta^\top\Delta\|_F\,\|\bar d\Omega\|_F
\le \|\Delta\|_F^2\,\|\bar d\Omega\|_F.
\)

Let
\(
x:=\|\bar d\Omega\|_F,
y:=\|\Delta\|_F.
\)
Then
\(
\phi''(0)=x^2+2y^2,
|\phi'''(0)|\le 2x^3+6xy^2.
\)
If $x=0$ the desired inequality is trivial. Otherwise, with
\(
u:=\frac{y^2}{x^2}\ge 0,
\)
we obtain
\(
\frac{|\phi'''(0)|}{2\{\phi''(0)\}^{3/2}}
\le
\frac{1+3u}{(1+2u)^{3/2}}.
\)
Define
\(
h(u):=(1+2u)^{3/2}-1-3u.
\)
Then $h(0)=0$ and
\(
h'(u)=3\bigl(\sqrt{1+2u}-1\bigr)\ge 0
\ \text{for all }u\ge 0,
\)
hence $h(u)\ge 0$ and therefore
\(
1+3u\le (1+2u)^{3/2}.
\)
This proves
\(
|\phi'''(0)|\le 2\{\phi''(0)\}^{3/2},
\)
i.e., $\LG$ is standard self-concordant.
\end{proof}

\section{Additional algorithmic details for the ADMM solvers}
\label{app:gap_framework}

This appendix collects the implementation material omitted from Section~\ref{sec:algorithms}. The ADMM splitting for the PN subproblem follows \citet{zhu2020convex}; we record the positive-definite $\Omega$-block proximal map and the primal-dual certificates used to stop the PN subproblem.

\subsection{The \texorpdfstring{$\Omega$}{Omega} block proximal operator}
\label{app:omega_prox}

This step follows the standard positive-definite $\ell_1$-penalized precision updates used in ADMM solvers; see, for example, \citet{xue2012positive}. We record it here because it is used repeatedly inside both the PG and PN implementations. We need to compute
\begin{equation}
\label{eq:omega_prox_primal}
\prox_{\kappa\|\cdot\|_{1,\mathrm{off}}+\iota_{\{\Omega\succeq \delta I\}}}(V)
=
\arg\min_{\Omega\succeq \delta I_q}
\Bigl\{\frac12\|\Omega-V\|_F^2+\kappa\|\Omega\|_{1,\mathrm{off}}\Bigr\},
\end{equation}
for symmetric $V\in\Symq$, threshold $\kappa>0$, and safeguard $\delta\ge 0$. Let \(\Omega^{\mathrm{soft}}=\operatorname{SoftThresh}_{\mathrm{off}}(V,\kappa),\) where only off-diagonal entries are soft-thresholded and the diagonal is left unchanged. If $\Omega^{\mathrm{soft}}\succeq \delta I_q$, then $\Omega^{\mathrm{soft}}$ is already the constrained proximal point, so no inner ADMM iterations are needed. In particular, at such updates the legacy and aggressive PG implementations are identical.

Otherwise, we enforce the positive-definite safeguard via ADMM with an auxiliary variable $K$:
\[
\min_{\Omega \in\Symq,K\succeq\delta I,\Omega=K}\;
\frac12\|\Omega-V\|_F^2+\kappa\|\Omega\|_{1,\mathrm{off}}.
\]
This problem can be solved exactly by Algorithm 1 in \citet{xue2012positive}.


\subsection{PN subproblem primal-dual gap}
\label{app:pn_gap}
We assume $H^{(m)}$ is positive definite. For the PN subproblem at outer iterate $m$, define the primal objective
\begin{equation}
\label{eq:pn_primal_objective}
Q_m(\zeta)
:=
\frac12 \tr\bigl(\zeta^\top H^{(m)}\zeta\,\Sigma^{(m)}\bigr)
+ \tr\bigl(\zeta^\top G^{(m)}\bigr)
+ \lambda_\gamma\|\gamma_\zeta\|_1
+ \lambda_\Omega\|\Omega_\zeta\|_{1,\mathrm{off}}
+ \iota_{\{\Omega_\zeta\succeq \delta I\}}.
\end{equation}
Consider any dual feasible certificate $(W_\gamma,Y_\Omega,S_\Omega)$ satisfying
\begin{equation}
\label{eq:pn_dual_feasible}
\|W_\gamma\|_{\max}\le \lambda_\gamma,
\qquad
Y_{\Omega,ii}=0,\ \ |Y_{\Omega,ij}|\le \lambda_\Omega\ (i\neq j),
\qquad
S_\Omega\preceq 0.
\end{equation}
Set $W_\Omega:=Y_\Omega+S_\Omega$ and $W:=(W_\gamma,W_\Omega)$. Fenchel duality gives the dual objective
\begin{equation}
\label{eq:pn_dual_obj}
D_m(W)
:=
-\frac12\tr\!\Bigl(
(W+G^{(m)})^\top\{H^{(m)}\}^{-1}(W+G^{(m)})\,\Omega^{(m)}
\Bigr)
-\delta\,\tr(S_\Omega).
\end{equation}
If $H^{(m)}$ is not positive definite, one could replace the inverse $\{H^{(m)}\}^{-1}$ by the Moore--Penrose inverse $\{H^{(m)}\}^{\dagger}$ in \eqref{eq:pn_dual_obj}.

Hence the primal-dual gap is
\begin{equation}
\label{eq:pn_gap}
\operatorname{gap}_{\mathrm{PN},m}(\zeta,W)
:=
Q_m(\zeta)-D_m(W)\ge 0,
\end{equation}
with relative version
\begin{equation}
\label{eq:pn_gap_rel}
\operatorname{gap}_{\mathrm{PN},m}^{\mathrm{rel}}(\zeta,W)
:=
\frac{\operatorname{gap}_{\mathrm{PN},m}(\zeta,W)}{1+|Q_m(\zeta)|}.
\end{equation}
In the ADMM implementation, the outer multiplier directly provides $W_\gamma=\rho U_\gamma$. For the $\Omega$ block, a box feasible matrix $Y_\Omega$ can be formed from the clipped multiplier in the nested routine of Appendix~\ref{app:omega_prox} and then combined with a semidefinite slack term $S_\Omega$ to form $W_\Omega=Y_\Omega+S_\Omega$. This yields an exact computable gap certificate for stopping and logging.

\section{Penalty path initialization at the null model}
\label{app:proof_null_corner}

\begin{proof}
By the KKT conditions, a point $(\gamma^\star,\Omega^\star)$ with $\Omega^\star\succ0$ minimizes \(F(\gamma,\Omega)=\LG(\gamma,\Omega)+\lambda_\gamma\|\gamma\|_1+\lambda_\Omega\|\Omega\|_{1,\mathrm{off}}\)
if and only if the zero matrix belongs to its subdifferential. 

We verify the KKT conditions at \(\gamma^\star=\mathbf 0,\ \Omega^\star=\diag(\SYY)^{-1}.\)
First, because $\gamma^\star=\mathbf 0$, \(\nabla_\gamma \LG(\gamma^\star,\Omega^\star)=-2\SXY.\) The subgradient condition for the entrywise $\ell_1$ penalty is therefore
\( 
\bigl\|\nabla_\gamma \LG(\gamma^\star,\Omega^\star)\bigr\|_{\max}
=
2\|\SXY\|_{\max}
\le \lambda_\gamma,
\)
which holds whenever $\lambda_\gamma\ge 2\|\SXY\|_{\max}=\lambda_\gamma^{\max}$.

Second, since $\gamma^\star=\mathbf 0$, \(\nabla_\Omega \LG(\gamma^\star,\Omega^\star)=\SYY-(\Omega^\star)^{-1}=\SYY-\diag(\SYY).\) The diagonal stationarity condition is satisfied exactly because $(\Omega^\star)^{-1}$ matches $\diag(\SYY)$, while the off-diagonal subgradient condition becomes \( \|\SYY\|_{\max,\mathrm{off}}\le \lambda_\Omega,\) which holds whenever $\lambda_\Omega\ge \|\SYY\|_{\max,\mathrm{off}}=\lambda_\Omega^{\max}$.
All KKT conditions are therefore satisfied. 
\end{proof}

\bibliographystyle{plainnat}
\bibliography{references_revised_clean_final}

\end{document}